\documentclass[11pt,a4paper]{article}
\usepackage{latexsym}
\usepackage{fullpage}
\usepackage{url}
\usepackage{xspace}
\usepackage{color}

\usepackage{tikz}

\newcommand{\eg}{e.g.\@\xspace}

\newcommand{\etal}{\textit{et~al.\@}\xspace}

\newcommand{\mpibarrier}{\texttt{MPI\_\-Barrier}\xspace}
\newcommand{\mpigather}{\texttt{MPI\_\-Gather}\xspace}
\newcommand{\mpigatherv}{\texttt{MPI\_\-Gatherv}\xspace}
\newcommand{\mpiscatter}{\texttt{MPI\_\-Scatter}\xspace}
\newcommand{\mpiscatterv}{\texttt{MPI\_\-Scatterv}\xspace}
\newcommand{\mpiallreduce}{\texttt{MPI\_\-Allreduce}\xspace}
\newcommand{\mpiint}{\texttt{MPI\_\-INT}\xspace}

\newcommand{\tuwgatherv}{\texttt{TUW\_\-Gatherv}\xspace}

\newcommand{\openmpi}{Open\,MPI\xspace}
\newcommand{\necmpi}{NEC\,MPI\xspace}
\newcommand{\mvapich}{MVAPICH\xspace}
\newcommand{\intelmpi}{Intel\,MPI\xspace}
\newcommand{\jupiteropenmpitwoone}{OpenMPI-2.0.1\xspace}
\newcommand{\jupitermvapichtwotwo}{MVAPICH2-2.2\xspace}
\newcommand{\jupiternecmpi}{NEC\,MPI-1.3.1\xspace}
\newcommand{\vscintelmpi}{Intel\,MPI 2017.2\xspace}
\newcommand{\gccversion}{gcc 4.9.2\xspace}

\newcommand{\intelmpiadjust}{\texttt{I\_\-MPI\_\-ADJUST}}

\newcommand{\guidelt}{\ensuremath{\preceq}\xspace}

\newcommand{\ceiling}[1]{\lceil #1\rceil}
\newcommand{\floor}[1]{\lfloor #1\rfloor}

\newtheorem{lemma}{Lemma}

\newtheorem{theorem}{Theorem}

\newenvironment{proof}{\emph{Proof:}}{$\Box$\newline}

\title{Practical, Linear-time, Fully Distributed Algorithms for
  Irregular Gather and Scatter\thanks{This work was in part supported
    by the Austrian FWF project ``Verifying self-consistent MPI
    performance guidelines'' (P25530).  The computational results
    presented have been achieved in part using the Vienna Scientific
    Cluster (VSC).}} 
\author{Jesper Larsson Tr\"aff\\
TU Wien\\
Faculty of Informatics, Institute of Information Systems\\ 
Research Group Parallel Computing\\ 
Favoritenstrasse 16/184-5, 1040 Wien, Austria}

\begin{document}
\maketitle

\begin{abstract}
We present new, simple, fully distributed, practical algorithms with
linear time communication cost for irregular gather and scatter
operations in which processors contribute or consume possibly
different amounts of data. In a linear cost transmission model with
start-up latency $\alpha$ and cost per unit $\beta$, the new
algorithms take time $3\ceiling{\log_2 p}\alpha+\beta \sum_{i\neq
  r}m_i$ where $p$ is the number of processors, $m_i$ the amount of
data for processor $i, 0\leq i<p$, and processor $r, 0\leq r<p$ a root
processor determined by the algorithm. For a fixed, externally given
root processor $r$, there is an additive penalty of at most
$\beta(M_{d'}-m_{r_{d'}}-\sum_{0\leq j<d'}M_j)$ time steps where each
$M_j$ is the total amount of data in a tree of $2^j$ different
processors with roots $r_j$ as constructed by the algorithm.  The
worst-case penalty is less than $\beta \sum_{i\neq r}m_i$ time steps.
The algorithms have attractive properties for implementing the
operations for MPI (the Message-Passing Interface). Standard
algorithms using fixed trees take time either $\ceiling{\log_2
  p}(\alpha+\beta \sum_{i\neq r} m_i)$ in the worst case, or
$\sum_{i\neq r}(\alpha+\beta m_i)$. We have used the new algorithms to
give prototype implementations for the \mpigatherv and \mpiscatterv
collectives of MPI, and present benchmark results from a small and a
medium-large InfiniBand cluster. In order to structure the
experimental evaluation we formulate new performance guidelines for
irregular collectives that can be used to assess the performance in
relation to the corresponding regular collectives. We show that the
new algorithms can fulfill these performance expectations with a large
margin, and that standard implementations do not.
\end{abstract}

\section{Introduction}

Gather and scatter operations are important collective operations for
collecting and distributing data among processors in a parallel system
with some chosen (and known) root processor, \eg, row-column
gather-scatter in linear algebra algorithms. The problems come in two
flavors, namely a \emph{regular} (or homogeneous) variant in which all
processors contribute or consume blocks of the same size, and an
\emph{irregular} (on inhomogeneous) variant in which the blocks may
have different sizes. For the irregular variant, the root processor may
or may not know the sizes of the blocks of data to be distributed to
or collected from the non-root processors.  While good algorithms and
implementations exist for different types of systems for the regular
problems, the irregular problems have been much less studied and often
only trivial algorithms with less than optimal performance (for small
to medium block sizes) are implemented. In this paper, we present new,
simple algorithms for the irregular gather and scatter problems with
many desirable properties for the practical implementation, and show
experimentally with implementations for and in MPI~\cite{MPI-3.0} that
they can perform much better and much more consistently than common
algorithms and implementations.

Gather and scatter operations are included as collective operations in
MPI in both variants~\cite[Chapter 5]{MPI-3.0}. For the regular
operations \mpigather and \mpiscatter, usually fixed (binomial) trees
are used (hierarchically) for short to medium sized blocks, while
large blocks are sent or received directly from or to the root. Since
the common block size is known, the MPI processes can consistently and
without any extra communication decide which algorithm to use.
Standard algorithms are surveyed by Chan
\etal~\cite{ChanHeimlichPurkayasthavandeGeijn07}, and analyzed under a
linear transmission cost model where they lead to optimal, linear
bandwidth, and optimal number of communication rounds (binomial
trees). Similar results for different communication networks were
presented early by Saad and Schulz~\cite{SaadSchultz89}. For the
irregular \mpigatherv and \mpiscatterv operations where only the root
process has full information on the sizes of the blocks contributed by
the other, non-root processes, the situation is different. Fixed
(oblivious) trees of logarithmic depth may lead to a large block being
sent a logarithmic number of times, and letting the non-root processes
send or receive directly from the root entails a linear number of
communication start-ups which might be too expensive when the non-root
blocks are small. Current MPI libraries, nevertheless, seem to use
variations of these algorithms. Tr\"aff~\cite{Traff04:gatscat} gave
algorithms specifically for MPI that rely on the global information on
block sizes available at the root process and use sorting to construct
good trees. These algorithms may therefore be too expensive when
non-root blocks are small. Variants of these algorithms were discussed
and benchmarked by Dichev
\etal~\cite{DichevRychovLastovetsky10}. Regular gather-scatter
problems for heterogeneous multiprocessors where communication links
may have different capabilities have been studied in several papers,
\eg, \cite{Ben-MiledFortesEigenmannTaylor98,HattaShibusa00}. These
algorithms also mostly rely on global knowledge (by one process) and
sorting by the transmission times between processes to construct good
communication schedules, but could be adopted to irregular
gather-scatter problems. Boxer and Miller~\cite{BoxerMiller04} study
the regular gather-scatter problems on the coarse grained
multiprocessor and concentrate on the problem of finding good spanning
trees for the machine in case. For hypercubes, compound scatter-gather
computations are studied more precisely by Charles and
Fraigniaud~\cite{CharlesFraigniaud93} who derive pipelined schedules
for the regular gather and scatter problems. Simple algorithms for the
regular problems in an asynchronous communication model that accounts
for delays and permits overlap were presented
in~\cite{ShibusawaMakinoNimiyaHatta00}.  Bhatt \etal
\cite{BhattPucciRanadeRosenberg93} study the irregular gather and
scatter problems in tree networks, and derive (nearly) optimal
schedules for arbitrary trees. This situation is somewhat orthogonal
to the usual objective of finding both a good spanning tree and a
corresponding schedule. The algorithms require full knowledge of the
message sequences to be scattered and gathered.

In the following we present new, simple algorithms for the irregular
gather and scatter problems with a number of desirable properties. For
the analysis, we assume a fully connected network with 1-ported,
bidirectional (telephone-like) communication.  We
let $p$ denote the number of processors which are numbered
consecutively from $0$ to $p-1$. For simplicity, we assume that the
cost of transmitting a message of $m$ units between any two processors
is linear and modeled as $\alpha+\beta m$, where $\alpha$ is a
communication start-up latency, and $\beta$ the transmission time per
unit. A processor involved in communication can start the next
transmission as soon as it has finished and selects from which other
processor to receive the next message.  In the gather and scatter
problems, each processor $i,0\leq i<p$ has a block of data of size
$m_i$ with $m_i\geq 0$ that it either wants to contribute to (gather)
or consume from (scatter) some root processor $r,0\leq r<p$. The root
$r$ is usually a given processor, and this $r$ is known to all other
processors. At the root, blocks are stored in processor order, that is
$m_0,m_1,m_2,\ldots, m_{p-1}$ (we assume that the root also has a
block $m_r$ which does not have to be transmitted). Any consecutive
sequence of blocks can be sent or received together as a single
message. Our algorithms do not assume that the root knows the size of
all $p$ data blocks, although the \mpigatherv and \mpiscatterv
operations do make this assumption and require this to be the case.

Our algorithms construct spanning trees of logarithmic depth, and need
only the optimal $\ceiling{\log_2 p}$ number of communication rounds
for the tree construction, each round consisting of at most two
communication steps. For the gathering or scattering of the data
blocks, another at most $\ceiling{\log_2 p}$ communication rounds are
needed (we present some practical improvements for large block
sizes). Trees are constructed in a distributed manner, with each
processor working only from gradually accumulated information, with no
dependence on global information (\eg, from the root) on the sizes of
all other data blocks.  The time for the root to gather or scatter all
data blocks from or to the non-root processors is linear, namely
$\ceiling{\log_2 p}\alpha+\beta \sum_{0\leq i<p, i\neq r} m_i$, with
an additive penalty of at most $\beta(M_{d'}-m_{r_{d'}}-\sum_{0\leq
  j<d'}M_j)$ time steps where each $M_j$ is the total amount of data in
a tree of $2^j$ different processors as constructed by the algorithm
for the case when the root is a fixed, externally given process (as in
\mpigatherv and \mpiscatterv). The worst-case penalty is less than
$\beta \sum_{i\neq r}m_i$ time steps. In contrast, for any fixed,
block-size oblivious binomial tree it is easy to construct a worst
case taking $\ceiling{\log_2 p}(\alpha+\beta \sum_{0\leq i<p, i\neq r}
m_i)$ time steps, namely by choosing $m_i=0$ for all processors except
one being farthest away for the root.  At all processors, blocks are
always sent and received in order: Any receive operation receives a
message consisting of blocks $m_k,m_{k+1},\ldots,m_{k+l}$. No,
potentially costly, local reordering of blocks in message buffers is
therefore necessary.

We have implemented our algorithms\footnote{The prototype
  implementations used here for evaluation are available.} to support
the \mpigatherv and \mpiscatterv operations, and evaluated them with
different block size distributions on a small InfiniBand cluster under
three different MPI libraries, and a medium-large InfiniBand cluster
under the vendor (Intel) MPI library. In order to structure the
comparison against the native MPI library implementations we formulate
expectations on the relative performance as new, self-consistent
performance guidelines~\cite{Traff16:autoguide,Traff10:selfcons}. We
can show that the new algorithms can in many situations significantly
outperform the native MPI library, and overall much better fulfill the
formalized performance expectations.

\section{Problem and algorithm}

We now present the algorithm for the irregular gather problem; the
scatter algorithm is analogous. Each of the $p$ processors has a data
block of $m_i$ units that it needs to contribute to some root
process $r, 0\leq r<p$. We organize the $p$ processes in a
$\ceiling{\log_2 p}$-dimensional (incomplete), ordered hypercube which
we use as a design vehicle, but communication can be between
processors that are not adjacent in the hypercube. We let $H_d, 0\leq
d\leq\ceiling{\log_2 p}$ denote a $d$-dimensional (incomplete)
hypercube consisting of (at most) $2^d$ processors. We say that the
hypercube $H_d$ is \emph{ordered} if the processors belonging to $H_d$
form a consecutive range
$[a2^d,\ldots,a2^d+2^d-1]=[a2^d,\ldots,(a+1)2^d-1]$ for
$a\in\{0,\ldots, \ceiling{p/2^d}-1\}$. The ordered hypercube $H_{d+1}$
consisting of processors $[a2^{d+1},\ldots,(a+1)2^{d+1}-1]$ is built
from two \emph{adjacent}, ordered hypercubes $H_d$ with processors
$[2a2^d,\ldots,(2a+1)2^d-1]$ and $[(2a+1)2^d,\ldots,(2a+2)2^d-1]$. If
$p$ is not a power of two, the last $H_d$ hypercube consists of the
processors $[(\ceiling{p/2^d}-1)2^d,\ldots,p-1]$.

By an \emph{ordered hypercube gather algorithm} for $H_d$ we mean
an algorithm for $H_d$ in which a processor in one of the subcubes
$H_{d-1}$ which has gathered all data from the processors of this subcube
sends all its data to a processor in the other subcube $H_{d-1}$ which
similarly has already gathered all data from that subcube. This
processor will now have gathered all data in the hypercube $H_d$
and will become the root processor of $H_d$.  Note that this may
require communicating along edges that do not belong to the hypercube,
but of course do belong to the fully connected network.

\begin{lemma}
\label{lem:existence}
For any $H_d$, there exists an ordered hypercube gather algorithm
that gathers the data to some root processor $r$ in $H_d$ in $d\alpha
+\beta \sum_{i\in H_d,i\neq r} m_i$ time units.
\end{lemma}

\begin{proof}
The claim follows by induction on $d$. For $H_0$ the sole processor
$r\in H_0$ already has the data $m_0$ and there is no further
cost. Let $H'_{d-1}$ and $H''_{d-1}$ be the two subcubes of $H_d$. By
the induction hypothesis there is a processor $r'$ of $H'_{d-1}$ that
has gathered all data of $H'_{d-1}$ in $t'=(d-1)\alpha +\beta
\sum_{i\in H'_{d-1},i\neq r'} m_i$ time steps, and a processor $r''$
that has gathered all data of $H''_{d-1}$ in $t''=(d-1)\alpha +\beta
\sum_{i\in H''_{d-1},i\neq r''} m_i$ time steps. Of the two root
processors $r'$ and $r''$, the one with the smaller gather time (with
ties broken in favor of the hypercube with the smallest amount of
data) sends its data to the other root processor. Say, $r'$ is the
root with $t'\leq t''$. Processor $r'$ sends a message of $\sum_{i\in
  H'_{d-1}}m_i$ units to root $r''$ which takes $\alpha +
\beta\sum_{i\in H'_{d-1}}m_i$ time steps. Adding to the time $t''$
already taken by the slower $r''$ to gather the data from $H''_{d-1}$ gives
$(d-1)\alpha +\beta \sum_{i\in H''_{d-1},i\neq r''} m_i+\alpha +
\beta\sum_{i\in H'_{d-1}}m_i = d\alpha +\beta \sum_{i\in H_d,i\neq r}
m_i$ as claimed. The root $r''$ of $H''_{d-1}$ becomes the root $r$ of
$H_d$.
\end{proof}

Since roots with smaller gather times sends to roots with larger
gather times, communication can readily take place with no delay for
the sending gather root processor to become ready. Since subcubes are
ordered, the data blocks received at a new root can easily be kept in
consecutive order. Note that for the gather times of the two roots
$r'$ and $r''$, $t'\leq t''$ if and only if $\sum_{i\in H'_{d-1},i\neq
  r'} m_i\leq\sum_{i\in H''_{d-1},i\neq r''} m_i$, so that the shape
of the constructed gather tree depends only on the block sizes and not
on the relative magnitudes of $\alpha$ and $\beta$. For each $H_d$
hypercube with root $r$, $\sum_{i\in H_{d-1},i\neq r'}m_i$ is an estimate
of the time to construct $H_d$.

\begin{lemma}
For any arbitrarily given root processor $r\in H_d$, there is an
ordered hypercube algorithm that gathers all data in $H_d$ to $r$ in
$d\alpha +\beta \sum_{i\in H_d,i\neq r} m_i$ time units with an
additive penalty of at most $\beta(M_d'-m_{r_{d'}}-\sum_{0\leq j<d'}M_j)$
time steps for some $d',d'<d$.  The root processor gathers data from
the roots in a sequence of ordered hypercubes
$H_0,H_1,\ldots,H_{d-1}$, each with a total amount of data $M_j$, and
$d'$ is the last such hypercube for which waiting time is incurred.
\end{lemma}

\begin{proof}
The construction of Lemma~\ref{lem:existence} is modified such that
data are always sent to processor $r$ if either $r'=r$ or $r''=r$. The
given root processor $r$ will therefore receive blocks from $d-1$
linear gather time subcubes $H_0, H_1, \ldots, H_{d-1}$. The amount of
data, and the time needed to gather the data in these $d$ hypercubes
is unrelated and may differ. Let $M_j=\sum_{i\in H_j}m_i$ be the
amount of data in hypercube $H_j$ with root processor $r_j$.  If the
time needed to gather the data in some $H_{d'}$ to $r_{d'}$, namely
$\alpha d'+\beta(M_{d'}-m_{r_{d'}})$, is larger than the time needed
to gather the data from the previous hypercubes $H_0,H_1,\ldots
H_{d'-1}$ to $r$, the root processor is delayed until the data gather
in $H_{d'}$ has completed. This delay is at most $\alpha
d'+\beta(M_{d'}-m_{r_{d'}})-(\alpha
d'+\beta(\sum_{j<d'}M_j)=\beta(M_{d'}-m_{r_{d'}}-\sum_{j<d'}M_j)$. Let
$d'$ be the last hypercube in the sequence incurring such a delay. The
total time to gather all data to the root $r$ is therefore $d\alpha
+\beta \sum_{i\in H_d,i\neq r} m_i$ plus the penalty of
$\beta(M_{d'}-m_{r_{d'}}-\sum_{j<d'}M_j)$.
\end{proof}

The resulting construction is easy to implement, and better than first
gathering to the linear time root determined by
Lemma~\ref{lem:existence} and then sending to the externally given
root $r$ which would incur an extra communication round and sending
the complete data $\sum_{i\in H_d}m_i$, effectively loosing half of
the communication bandwidth (although still being linear).

\begin{figure}
\begin{center}
\begin{tikzpicture}
\draw (0.25,0.25) rectangle(3.25,3.25); 
\draw (3.5,0.25) rectangle(6.5,3.25);
\draw (0,0) rectangle(6.75,6.75);
\draw (0.75,2) node {$H'_d$};
\draw (4,2) node {$H''_d$};
\draw (0.75,4.25) node {$H_{d+1}$};
\draw [->] (3,4.25) -- (6.25,4.25) node[pos=0.5,above]{$(\sum_{i\neq r'\in H'_d} m_i,m_{r'},r')$};
\draw [<-] (3,4) -- (6.25,4) node[pos=0.5,below]{$(\sum_{i\neq r''\in H''_d} m_i,m_{r''},r'')$};
\draw [->] (3,6) -- (0.75,6) node[left] {$r'$};
\draw [->] (6.25,6) -- (4.5,6) node[left] {$r''$};
\end{tikzpicture}
\end{center}
\caption{An iteration of the algorithm of Lemma~\ref{lem:construction}
  showing the communication necessary to join two adjacent, ordered
  hypercubes $H'_d$ and $H''_d$ into the larger hypercube
  $H_{d+1}$. The fixed roots first exchange information on the gather
  times, the size of the root data blocks, and the identity of the
  gather roots in the respective subcubes. In the next step, the
  gather roots $r'$ and $r''$ receive this information from their
  fixed roots, so that they can consistently determine which will be
  the gather root for $H_{d+1}$.}
\label{fig:algorithm}
\end{figure}
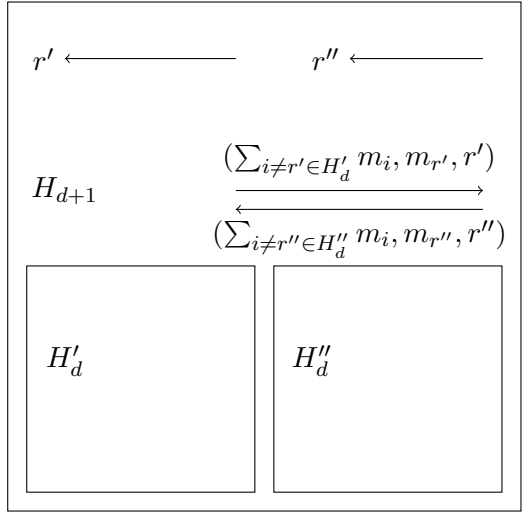

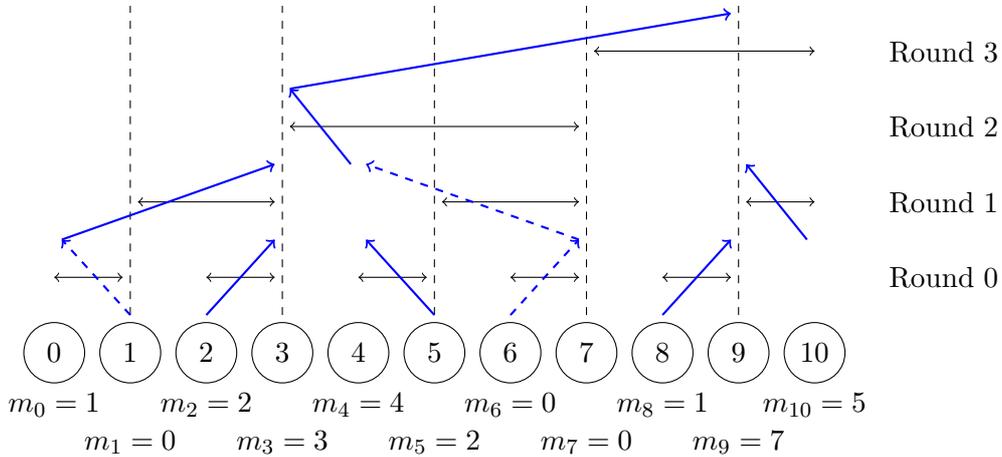
\begin{figure}
\begin{center}
\begin{tikzpicture}
\foreach \y in {0,1,2,3} 
  \draw (11.7,\y+1) node {Round $\y$};
\foreach \x in {0,...,10} 
  \draw (\x,0) circle (0.4) node {$\x$};
\foreach \x / \m in {0/1,2/2,4/4,6/0,8/1,10/5} 
  \draw (\x,-0.4) node[below] {$m_{\x}=\m$};
\foreach \x / \m in {1/0,3/3,5/2,7/0,9/7} 
  \draw (\x,-0.9) node[below] {$m_{\x}=\m$};
\foreach \x in {0,2,...,8} 
  \draw [<->] (\x,1) -- (\x+0.9,1);
\foreach \x in {1,5,...,10} 
  \draw [<->] (\x+0.1,2) -- ({min(\x+1.9,10)},2);
\foreach \x in {3} 
  \draw [<->] (\x+0.1,3) -- ({min(\x+3.9,10)},3);
\draw [<->] (7.1,4) -- (10,4);
\foreach \x in {1,3,...,9} 
  \draw [dashed,thin] (\x,0.5) -- (\x,4.6);
\draw [->,blue,thick,dashed] (1,0.5) -- (0.1,1.5);
\draw [->,blue,thick] (2,0.5) -- (2.9,1.5);
\draw [->,blue,thick] (5,0.5) -- (4.1,1.5);
\draw [->,blue,thick,dashed] (6,0.5) -- (6.9,1.5);
\draw [->,blue,thick] (8,0.5) -- (8.9,1.5);
\draw [->,blue,thick] (0.1,1.5) -- (2.9,2.5);
\draw [->,blue,thick,dashed] (6.9,1.5) -- (4.1,2.5);
\draw [->,blue,thick] (9.9,1.5) -- (9.1,2.5);
\draw [->,blue,thick] (3.9,2.5) -- (3.1,3.5);
\draw [->,blue,thick] (3.1,3.5) -- (8.9,4.5);
\end{tikzpicture}
\end{center}
\caption{A linear-time, ordered gather tree for $p=11$ processors and
  root $9$ with the indicated block sizes $m_i$ as constructed by the
  algorithm of Lemma~\ref{lem:construction}. Thick (blue) arrows are the
  gather tree edges with dotted arrows indicating data sizes of zero
  with no actual communication. Thin arrows indicate the exchange
  between fixed roots as needed to construct the ordered gather tree.}
\label{fig:example}
\end{figure}

The communication structure of an ordered hypercube gather algorithm
is a binomial tree with a particular numbering of the tree roots
determined by the $p$ data block sizes $m_i, 0\leq i<p$. An example is
shown in Figure~\ref{fig:example}. This tree can
be constructed efficiently as shown by the next lemma which is 
illustrated in Figure~\ref{fig:algorithm}.

\begin{lemma}
\label{lem:construction}
For any $H_d$, the gather communication tree can be constructed in $d$
communication rounds, each comprising at most two send and receive
operations.
\end{lemma}

\begin{proof}
The communication tree is constructed iteratively, maintaining the
following invariant. Each $H_d$ has a predetermined, fixed root that
can readily be computed by any processor, and a
gather root $r$ which will gather the data from $H_d$ as per
Lemma~\ref{lem:existence}. The fixed root and the gather root are not
necessarily distinct processors. Both the fixed and the gather root
processors know that they have this role and which processor has the
other role, and each knows the total amount of data in $H_d$. When the
hypercube $H_{d+1}$ is formed from $H'_d$ and $H''_d$, the fixed root
of $H'_d$ knows which processor is the fixed root of $H''_d$ and vice
versa. The gather roots do not known the gather root of the other
subcube.

For all $H_0$ subcubes the invariant holds with fixed and gather root
being the sole processor in $H_0$. To maintain the invariant for
$H_{d+1}$, the two fixed roots of the $H_d$ subcubes exchange
information on their estimated gather time, the size of the root data
blocks, and the identity of the gather root processors.  Both fixed
roots can now determine which gather root will be the gather root of
$H_{d+1}$, namely the gather root of the subcube with the largest
gather time estimate $\sum_{i\in H_d,i\neq r}m_i$ (with ties broken
arbitrarily, but consistently). The first time a fixed root of some
$H_d$ by the exchange determines that it will become a gather root of
$H_{d+1}$, this new gather root knows that it is a gather root. To
maintain the invariant for the following iterations, if the gather
root of $H_d$ does not know whether it will be the gather root of
$H_{d+1}$, it receives information on the gather root in $H_{d+1}$
from its fixed root in $H_d$ which per invariant knows the identity of
the gather root in $H_d$. By exchanging both the gather times
$\sum_{i\in H_d,i\neq r}m_i$ and the sizes of the root data blocks $m_r$,
gather roots can compute the amount of data to be received in each
communication round.

The construction takes $\ceiling{\log_2 p}$ iterations in each of
which pairwise exchanges between the fixed roots of adjacent
hypercubes take place. After this, at most one transmission between
each fixed root and its corresponding gather root is necessary, except
for the first communication round where no such transmission is
needed. Thus at most $2\ceiling{\log_2 p}-1$ dependent communication
operations are required. All information exchanged is of constant
size, consisting of the gather time, the size of the root data block,
and the identity of the gather root.
\end{proof}

As fixed root for a subcube $H_d$ consisting of processors
$[a2^d,\ldots,(a+1)2^d-1]$ we can choose, \eg, the last processor
$i=(a+1)2^d-1$. For the fixed root of this $H_d$ to find the fixed
root of its adjacent subcube, the $d$'th bit of $i$ has to be flipped.
If the $d$'th bit is a $1$, the processor will survive as the fixed
root of $H_{d+1}$. If the number of processors $p$ is not a power of
two, only the last processor $p-1$ has to be specially treated. If a
fixed root in iteration $d$, by flipping bit $d$, determines that its
partner fixed root is larger than $p-1$ it instead chooses $p-1$ as
fixed root in its adjacent, incomplete hypercube. In this iteration,
processor $p-1$ must be prepared to act as fixed root. In iteration
$d$, if processor $p-1$ has bit $d$ set, it knows that some lower
numbered fixed root has chosen $p-1$ as fixed root, and this adjacent
fixed root is $(\ceiling{p/2^d}-1)2^d-1$.  Otherwise, if bit $d$ is
not set, processor $p-1$ has no role in iteration $d$.

Together, these remarks and the three lemmas give the main result.

\begin{theorem}
For any number of processors $p$, the irregular gather problem with
root $r$ and block size $m_i$ for processor $i, 0\leq i<p$ can be
solved in at most $3\ceiling{\log_2 p}\alpha+\beta \sum_{0\leq
  i<p,i\neq r} m_i$ communication time steps with an additive penalty
of at most $\beta(M_d'-m_{r_d'}-\sum_{0\leq j<d'}M_j)$ time steps,
each $M_j$ being the total amount of data in a tree of $2^j$ distinct
processors with local root $r_j$.
\end{theorem}

The linear time gather trees can likewise be used for the irregular
scatter problem.  Also, both tree construction and communication
algorithms can obviously be extended to $k$-ported communication
systems, which reduces the number of communication rounds needed from
$\ceiling{\log_2 p}$ to $\ceiling{\log_{k+1} p}$. It is perhaps worth
pointing out that the constructions also provide ordered communication
trees for the regular gather and scatter (as well as for
reduction-to-root) operations with the optimal $\ceiling{\log_2 p}$
number of communication rounds. If all processors know the common root
$r$, tree construction can be done without any actual, extra
communication.

\section{MPI implementations}

We have implemented the irregular gather algorithm in MPI with the
same interface as the \mpigatherv operation. We can thus readily
compare our \tuwgatherv implementation against \mpigatherv. We use the
algorithm of Lemma~\ref{lem:construction} to construct a gather
communication tree which we represent at each process as a sequence of
receive operations followed by a send operation. We use non-blocking
receives to better absorb delays by some MPI processes finishing late,
such that the reception order is determined by the times the processes
become ready. For non-root processes, intermediate buffers gather data
from the processes' children. Since the sizes of all received data are
known by construction, and since it can be assumed that all children
send rank ordered data blocks, it is easy to keep blocks stored in
intermediate buffers in rank order. Since all processes in the
\mpigatherv operation must supply an MPI datatype describing the types
and structure of their blocks, and all data blocks must eventually
match the datatype supplied by the root process, it is possible to
receive and send all intermediate data blocks with a correct MPI
derived datatype. To this end, the signature datatype described
in~\cite{Traff16:typeprog} can be used. Since blocks can be described
by different types with different counts by different processes, it is
important that the signature type used is a ``smallest common
block''. At the root, the gathered data blocks must eventually be
stored as described by the list of displacements and block sizes
supplied in the root process' call of \mpigatherv. This can be
accomplished by constructing a corresponding indexed derived datatype
for each of the children describing where the data blocks go. No
explicit, intermediate buffering at the root is therefore necessary,
and in that sense a zero-copy implementation of the gather algorithm
is possible. If the root displacements describe a contiguous segment
of blocks in rank order (as may be the case in applications), no such
datatype is necessary, and the blocks can be received directly into
their correct positions in the root receive buffer. Our prototype
implementation works under this assumption.

Despite the linear gather time guaranteed by the algorithm, sending
large data blocks multiple times through the gather tree incurs
unnecessary, repeated transmission costs. Practical performance may be
better if such large blocks are sent directly to the root process. We
could implement graceful degradation behavior~\cite{Traff04:gatscat}
by introducing a gather subtree threshold beyond which a subtree in
the gather tree shall send its data directly to the root. In order for
this to work, the tree construction algorithm is extended to also
count, for each gather root, the number of subtrees that have exceeded
the threshold. With this information, the root process knows how many
subtrees will send data blocks directly to the root. For each of these
subtrees, the root first needs to receive the size of the blocks,
based on which the correct position of the block in the receive buffer
can be computed and the block received. We have not implemented this
potential improvement.


\section{A padding performance guideline for irregular collectives}
\label{sec:newguidelines}

Self-consistent MPI performance guidelines formalize expectations on
the performance of given MPI operations by relating them to the
performance of other MPI operations implementing the same
functionality~\cite{Traff10:selfcons}. If a performance guideline is
violated, it gives a constructive hint to the application programmer
and the MPI library implementer how the given operation can be
improved in the given context. Performance guidelines thus provide
sanity checks for MPI library implementations, and can be helpful in
structuring experiments~\cite{Traff17:expected,Traff16:autoguide}.

In order to use regular collectives correctly, the application
programmer must know that all processes supply the same data sizes and
each process must know this data size. The irregular collectives have
a weaker precondition: It suffices that each process by itself knows
its data size with the only requirement that processes that pairwise
exchange data must know and supply the same sizes. If an irregular
collective is used in a situation where a regular one could have been
used instead, we would expect the regular collective to perform
better, or at least not worse in that situation. This is captured in
the performance guideline for \mpigatherv below.
\begin{eqnarray}
\mpigather(m) & \guidelt & \mpigatherv(m) 
\label{guide:regirreg}
\end{eqnarray}
Here $m$ is the total amount of data to be gathered at the root
process, and the guideline states that in a situation where \mpigather
can be used $(m_i=m/p)$, this should perform at least as well as using
instead \mpigatherv, all other things (\eg, root process) being equal
in the two sides of the equation. If the guideline is violated, which
can be tested experimentally, there is something wrong with the MPI
library, and the user would do better by using \mpigatherv instead of
\mpigather. There are reasons to expect that the guideline is not
violated. The \mpigather operation is more specific, does not take long
argument lists of counts and displacements, and good, tree-based
algorithms exist and may have been implemented for this operation.

A common way of dealing with slightly irregular problems is to
transform them into regular ones by padding all buffers up to some
common size and solving the problem by a corresponding regular
collective operation. The argument for having the specialized,
irregular collectives in the MPI specification is that a library can
possibly do better than (or at least as good as) this manual
solution. Thus, we would like to expect that \mpigatherv performs no
worse than first agreeing on the common buffer size and then doing the
regular collective on this, possibly larger common size. This is
expressed in the second irregular performance guideline.
\begin{eqnarray}
\mpigatherv(m) & \guidelt & \mpiallreduce(1) + \mpigather(m') 
\label{guide:padding}
\end{eqnarray}
where $m'=p\max_{0\leq i<p}m_i$ is the total amount of data to be
gathered by the regular \mpigather as computed by the \mpiallreduce
operation. Again, if experiments show this guideline violated, there
is an immediate hint for the application programmer on how to do
better: Use padding. Here we assume that the application programmer
can organize his padded buffers such that no copying back and forth
between buffers is necessary; this may not always be possible, so the
guideline should not be interpreted too strictly but allow some extra
slack on the right-hand side upper bound. Nevertheless, it constrains
what should be expected by a good implementation of \mpigatherv.

The second guideline is particularly interesting for the regular case
where $m_i=m/p$. Here it says that the overhead for \mpigatherv
compared to \mpigather should not be more than a single, small
\mpiallreduce operation. This may be difficult for MPI libraries to
satisfy, but indeed, if it is not, the usefulness of \mpigatherv may
be questionable.

The two performance guidelines give less trivial performance
expectations against which to test our new algorithms instead of only
comparing to the \mpigatherv and \mpigather implementations in some
given MPI library.

\section{Experiments}


We give a preliminary evaluation of the \tuwgatherv implementation. We
do this by comparing to \mpigather and \mpigatherv guided by the
performance guidelines explained in Section~\ref{sec:newguidelines}.


\begin{table}
\caption{Results for \necmpi, $p=35\times 16=560$. Running times are
  in microseconds ($\mu s$). Numbers in red, bold font show violations
  of the performance guidelines, Guideline~(\ref{guide:regirreg}) in the
  \mpigather column, Guideline~(\ref{guide:padding}) in column
  \mpigatherv}
\label{tab:necmpi}
\begin{tiny}
\begin{tabular}{crrrrrrrrrr}
Problem & $m$ & $m'$ & 
\multicolumn{2}{c}{\mpigather} & 
\multicolumn{2}{c}{Guideline~(\ref{guide:padding})} & 
\multicolumn{2}{c}{\mpigatherv} &
\multicolumn{2}{c}{\tuwgatherv} \\
& & & 
\multicolumn{1}{c}{avg} & \multicolumn{1}{c}{min} & 
\multicolumn{1}{c}{avg} & \multicolumn{1}{c}{min} & 
\multicolumn{1}{c}{avg} & \multicolumn{1}{c}{min} & 
\multicolumn{1}{c}{avg} & \multicolumn{1}{c}{min} \\
\hline
Same
& 560 & 560 & 337.83 & 18.12 & 59.43 & 36.95 & 1755.53 & \color{red}{\textbf{1194.95}} & 50.20 & 20.98\\
& 5600 & 5600 & 49.69 & 33.86 & 76.92 & 57.94 & 1753.57 & \color{red}{\textbf{1235.96}} & 45.06 & 37.91\\
& 56000 & 56000 & 183.12 & 169.99 & 207.12 & 186.92 & 1810.53 & \color{red}{\textbf{1347.06}} & 138.28 & 119.92\\
& 560000 & 560000 & 1307.84 & 1293.18 & 1336.58 & 1312.02 & 2644.22 & \color{red}{\textbf{2297.16}} & 824.08 & 802.04\\
& 5600000 & 5600000 & 16874.53 & \color{red}{\textbf{9708.17}} & 9793.19 & 9739.16 & 9243.09 & 7848.02 & 7972.50 & 7951.97\\
Random
& 844 & 1120 & 31.67 & 20.98 & 55.82 & 37.91 & 1707.73 & \color{red}{\textbf{1415.01}} & 62.38 & 28.85\\
& 5989 & 11200 & 68.23 & 57.94 & 93.61 & 77.01 & 1720.36 & \color{red}{\textbf{948.91}} & 88.62 & 42.92\\
& 57615 & 112000 & 302.17 & 290.16 & 324.11 & 305.18 & 1851.97 & \color{red}{\textbf{1305.82}} & 185.11 & 144.96\\
& 571327 & 1119440 & 2243.32 & 2229.93 & 2258.11 & 2242.09 & 2702.19 & \color{red}{\textbf{2511.02}} & 1054.04 & 827.07\\
& 5546939 & 11189360 & 21490.23 & 18246.89 & 31686.29 & 18262.15 & 7940.00 & 7556.92 & 15855.80 & 9074.93\\
Spikes
& 1036 & 2800 & 38.74 & 29.09 & 59.58 & 43.87 & 5262.01 & \color{red}{\textbf{1353.03}} & 58.94 & 28.85\\
& 6391 & 28000 & 111.11 & 98.94 & 134.16 & 118.97 & 1726.01 & \color{red}{\textbf{1243.11}} & 76.38 & 40.05\\
& 57945 & 280000 & 684.50 & 667.10 & 703.41 & 684.98 & 1759.52 & \color{red}{\textbf{1068.12}} & 163.28 & 113.96\\
& 570446 & 2800000 & 4839.81 & 4822.02 & 4874.30 & 4847.05 & 2029.88 & 1711.85 & 1088.44 & 1050.00\\
& 6100438 & 28000000 & 46745.77 & 44775.96 & 50961.33 & 44814.11 & 8909.78 & 8166.07 & 10682.46 & 8350.13\\
Decreasing
& 842 & 1680 & 33.10 & 25.99 & 53.16 & 41.96 & 1690.14 & \color{red}{\textbf{845.91}} & 40.90 & 25.03\\
& 5900 & 11760 & 70.28 & 61.99 & 92.52 & 79.87 & 1731.86 & \color{red}{\textbf{952.96}} & 94.30 & 61.04\\
& 56400 & 112560 & 315.14 & 299.93 & 341.26 & 318.05 & 1986.37 & \color{red}{\textbf{1213.07}} & 212.21 & 174.05\\
& 561320 & 1120560 & 2244.04 & 2228.98 & 2267.95 & 2247.10 & 2900.91 & \color{red}{\textbf{2758.03}} & 1223.44 & 1196.86\\
& 5610320 & 11200560 & 22649.90 & 18422.13 & 18473.74 & 18440.96 & 7727.55 & 7624.86 & 12091.68 & 11695.86\\
Alternating
& 560 & 560 & 5793.04 & 20.03 & 48.54 & 34.81 & 4837.82 & \color{red}{\textbf{1188.04}} & 32.95 & 20.98\\
& 5600 & 8400 & 58.68 & 49.83 & 86.03 & 67.00 & 1720.53 & \color{red}{\textbf{1029.01}} & 64.41 & 41.01\\
& 56000 & 84000 & 244.91 & 231.03 & 269.17 & 246.05 & 1860.01 & \color{red}{\textbf{1120.09}} & 153.70 & 139.00\\
& 560000 & 840000 & 1649.08 & 1631.98 & 1666.12 & 1648.90 & 11386.20 & \color{red}{\textbf{2473.12}} & 980.67 & 967.03\\
& 5600000 & 8400000 & 29599.52 & 13819.93 & 22029.07 & 13828.04 & 7947.92 & 7896.18 & 14881.42 & 9907.96\\
Two blocks
& 2 & 560 & 28.39 & 18.12 & 50.70 & 36.95 & 5.74 & 1.91 & 28.09 & 18.12\\
& 20 & 5600 & 49.09 & 41.01 & 72.24 & 58.89 & 6.89 & 2.86 & 28.25 & 18.12\\
& 200 & 56000 & 185.02 & 171.90 & 209.54 & 189.07 & 7.31 & 2.86 & 29.83 & 19.07\\
& 2000 & 560000 & 1306.06 & 1291.04 & 1331.58 & 1313.92 & 12.96 & 8.11 & 33.93 & 23.13\\
& 20000 & 5600000 & 16682.08 & 9497.17 & 9559.33 & 9510.04 & 39.24 & 36.00 & 63.39 & 51.02\\
\hline
\end{tabular}
\end{tiny}
\end{table}

\begin{table}
\caption{Results for \mvapich, $p=35\times 16=560$. Running times are
  in microseconds ($\mu s$). Numbers in red, bold font show violations
  of the performance guidelines, Guideline~(\ref{guide:regirreg}) in the
  \mpigather column, Guideline~(\ref{guide:padding}) in column
  \mpigatherv}
\label{tab:mvapich}
\begin{tiny}
\begin{tabular}{crrrrrrrrrr}
Problem & $m$ & $m'$ & 
\multicolumn{2}{c}{\mpigather} & 
\multicolumn{2}{c}{Guideline~(\ref{guide:padding})} & 
\multicolumn{2}{c}{\mpigatherv} &
\multicolumn{2}{c}{\tuwgatherv} \\
& & & 
\multicolumn{1}{c}{avg} & \multicolumn{1}{c}{min} & 
\multicolumn{1}{c}{avg} & \multicolumn{1}{c}{min} & 
\multicolumn{1}{c}{avg} & \multicolumn{1}{c}{min} & 
\multicolumn{1}{c}{avg} & \multicolumn{1}{c}{min} \\
\hline
Same
& 560 & 560 & 56.69 & 25.03 & 295.59 & 77.01 & 1083.19 & \color{red}{\textbf{887.87}} & 53.32 & 38.15\\
& 5600 & 5600 & 250.77 & 84.88 & 284.61 & 114.92 & 1229.84 & \color{red}{\textbf{1013.99}} & 77.34 & 48.88\\
& 56000 & 56000 & 317.90 & 302.08 & 366.31 & 323.06 & 1403.33 & \color{red}{\textbf{1122.95}} & 193.72 & 155.93\\
& 560000 & 560000 & 3736.72 & \color{red}{\textbf{3597.02}} & 3852.69 & 3761.05 & 3760.52 & 3576.99 & 1149.36 & 1132.01\\
& 5600000 & 5600000 & 31109.11 & \color{red}{\textbf{30703.07}} & 30893.50 & 30778.17 & 7758.92 & 7701.16 & 15844.17 & 15782.12\\
Random
& 844 & 1120 & 84.15 & 76.06 & 105.57 & 93.94 & 2357.83 & \color{red}{\textbf{1780.03}} & 60.69 & 49.83\\
& 5989 & 11200 & 141.02 & 136.14 & 161.29 & 150.92 & 2413.02 & \color{red}{\textbf{1760.01}} & 71.50 & 56.98\\
& 57615 & 112000 & 573.50 & 535.01 & 608.74 & 573.87 & 2597.00 & \color{red}{\textbf{1867.06}} & 168.84 & 157.12\\
& 571327 & 1119440 & 5006.39 & 4858.97 & 7115.21 & 5001.07 & 3812.06 & 3521.92 & 1856.02 & 849.96\\
& 5546939 & 11189360 & 54121.13 & 53570.03 & 54998.99 & 52770.14 & 8087.99 & 7569.07 & 13339.92 & 13207.91\\
Spikes
& 1036 & 2800 & 98.12 & 81.06 & 117.19 & 102.04 & 2353.78 & \color{red}{\textbf{1659.87}} & 57.43 & 47.92\\
& 6391 & 28000 & 204.39 & 198.84 & 226.27 & 215.05 & 2351.88 & \color{red}{\textbf{1694.92}} & 70.14 & 56.98\\
& 57945 & 280000 & 1200.79 & 1181.13 & 1225.05 & 1204.01 & 2432.12 & \color{red}{\textbf{1807.93}} & 157.08 & 144.96\\
& 570446 & 2800000 & 5283.16 & 4540.92 & 5573.90 & 5078.08 & 2965.85 & 2594.95 & 1153.44 & 1137.97\\
& 6100438 & 28000000 & 130602.19 & 129426.96 & 131404.61 & 129746.20 & 8642.77 & 8414.03 & 13026.98 & 12981.18\\
Decreasing
& 842 & 1680 & 85.32 & 72.96 & 110.13 & 96.80 & 2349.97 & \color{red}{\textbf{1697.78}} & 53.90 & 44.82\\
& 5900 & 11760 & 144.53 & 133.99 & 165.36 & 154.97 & 2398.28 & \color{red}{\textbf{1858.00}} & 86.67 & 75.10\\
& 56400 & 112560 & 580.63 & 550.03 & 607.42 & 577.93 & 2753.80 & \color{red}{\textbf{1947.16}} & 235.09 & 220.06\\
& 561320 & 1120560 & 5091.47 & 4953.86 & 5138.67 & 5054.95 & 4221.21 & 3954.17 & 4238.99 & 2449.99\\
& 5610320 & 11200560 & 54855.09 & 54377.08 & 55570.82 & 54203.99 & 7769.08 & 7728.82 & 19093.30 & 19005.78\\
Alternating
& 560 & 560 & 84.05 & 57.94 & 107.30 & 91.79 & 2298.24 & \color{red}{\textbf{1638.17}} & 52.58 & 42.92\\
& 5600 & 8400 & 127.73 & 115.87 & 152.87 & 139.95 & 2327.04 & \color{red}{\textbf{1632.93}} & 69.52 & 56.98\\
& 56000 & 84000 & 441.42 & 401.02 & 469.97 & 445.84 & 2538.67 & \color{red}{\textbf{2073.05}} & 172.13 & 157.12\\
& 560000 & 840000 & 4525.85 & 4410.98 & 4591.90 & 4480.84 & 3777.33 & 3542.90 & 1132.91 & 1110.79\\
& 5600000 & 8400000 & 52944.41 & 46571.02 & 46649.89 & 46545.03 & 11048.42 & 7767.92 & 15145.07 & 15044.93\\
Two blocks
& 2 & 560 & 81.62 & 73.19 & 279.81 & 92.03 & 7.70 & 1.91 & 40.42 & 30.99\\
& 20 & 5600 & 613.77 & 108.00 & 197.26 & 126.12 & 7.39 & 1.91 & 40.66 & 30.99\\
& 200 & 56000 & 787.44 & 301.84 & 477.54 & 332.83 & 8.76 & 1.91 & 82.04 & 30.99\\
& 2000 & 560000 & 3760.01 & 3598.93 & 3855.75 & 3764.87 & 11.48 & 2.86 & 45.08 & 33.86\\
& 20000 & 5600000 & 32092.41 & 31884.91 & 31959.97 & 31888.96 & 41.83 & 36.00 & 74.30 & 62.94\\
\hline
\end{tabular}
\end{tiny}
\end{table}

\begin{table}
\caption{Results for \openmpi, $p=35\times 16=560$. Running times are
  in microseconds ($\mu s$). Numbers in red, bold font show violations
  of the performance guidelines, Guideline~(\ref{guide:regirreg}) in the
  \mpigather column, Guideline~(\ref{guide:padding}) in column
  \mpigatherv}
\label{tab:openmpi}
\begin{tiny}
\begin{tabular}{crrrrrrrrrr}
Problem & $m$ & $m'$ & 
\multicolumn{2}{c}{\mpigather} & 
\multicolumn{2}{c}{Guideline~(\ref{guide:padding})} & 
\multicolumn{2}{c}{\mpigatherv} &
\multicolumn{2}{c}{\tuwgatherv} \\
& & & 
\multicolumn{1}{c}{avg} & \multicolumn{1}{c}{min} & 
\multicolumn{1}{c}{avg} & \multicolumn{1}{c}{min} & 
\multicolumn{1}{c}{avg} & \multicolumn{1}{c}{min} & 
\multicolumn{1}{c}{avg} & \multicolumn{1}{c}{min} \\
\hline
Same
& 560 & 560 & 173.80 & 45.00 & 239.85 & 123.00 & 967.65 & \color{red}{\textbf{170.00}} & 219.27 & 67.00\\
& 5600 & 5600 & 156.08 & 74.00 & 242.00 & 153.00 & 969.65 & \color{red}{\textbf{176.00}} & 124.20 & 79.00\\
& 56000 & 56000 & 926.56 & \color{red}{\textbf{764.00}} & 1065.35 & 872.00 & 430.83 & 191.00 & 749.55 & 535.00\\
& 560000 & 560000 & 3619.15 & \color{red}{\textbf{3378.00}} & 3758.16 & 3463.00 & 1788.52 & 1006.00 & 2390.60 & 1899.00\\
& 5600000 & 5600000 & 14500.19 & 13844.00 & 14405.19 & 14053.00 & 14210.47 & 13993.00 & 16043.45 & 15939.00\\
Random
& 843 & 1120 & 106.56 & 61.00 & 174.33 & 125.00 & 1139.69 & 115.00 & 121.85 & 70.00\\
& 5892 & 11200 & 155.07 & 90.00 & 207.55 & 160.00 & 1073.25 & 134.00 & 132.68 & 89.00\\
& 57435 & 112000 & 1119.33 & 1071.00 & 1238.20 & 1184.00 & 485.49 & 193.00 & 560.53 & 509.00\\
& 542098 & 1116640 & 7284.08 & 7059.00 & 7526.65 & 7228.00 & 2214.36 & 1581.00 & 1890.25 & 1842.00\\
& 5687094 & 11173120 & 20778.19 & 19741.00 & 20860.49 & 20379.00 & 14068.03 & 13896.00 & 12936.69 & 12858.00\\
Spikes
& 984 & 2800 & 108.67 & 62.00 & 187.49 & 137.00 & 1125.93 & \color{red}{\textbf{144.00}} & 125.40 & 70.00\\
& 6146 & 28000 & 209.33 & 127.00 & 256.87 & 196.00 & 1056.89 & \color{red}{\textbf{214.00}} & 132.23 & 77.00\\
& 54951 & 280000 & 1885.92 & 1829.00 & 2004.80 & 1932.00 & 1142.00 & 285.00 & 594.28 & 444.00\\
& 525455 & 2800000 & 10698.63 & 10047.00 & 10799.77 & 10433.00 & 3012.12 & 2909.00 & 1764.33 & 1707.00\\
& 5000460 & 28000000 & 41352.32 & 40729.00 & 42246.07 & 41248.00 & 10434.21 & 10021.00 & 10070.91 & 9774.00\\
Decreasing
& 842 & 1680 & 106.77 & 56.00 & 181.56 & 134.00 & 1107.64 & 126.00 & 123.28 & 69.00\\
& 5900 & 11760 & 143.09 & 96.00 & 421.40 & 161.00 & 1258.08 & \color{red}{\textbf{518.00}} & 162.55 & 106.00\\
& 56400 & 112560 & 1106.97 & 1058.00 & 1233.16 & 1158.00 & 451.28 & 195.00 & 638.89 & 589.00\\
& 561320 & 1120560 & 7348.33 & 7125.00 & 7994.61 & 7220.00 & 2585.16 & 2077.00 & 2586.47 & 2427.00\\
& 5610320 & 11200560 & 20940.69 & 20003.00 & 20914.79 & 20360.00 & 13757.08 & 13538.00 & 26159.17 & 21454.00\\
Alternating
& 560 & 560 & 137.84 & 58.00 & 245.55 & 124.00 & 1209.64 & \color{red}{\textbf{136.00}} & 120.92 & 63.00\\
& 5600 & 8400 & 142.93 & 94.00 & 201.03 & 155.00 & 1204.63 & \color{red}{\textbf{282.00}} & 115.28 & 71.00\\
& 56000 & 84000 & 935.68 & 886.00 & 1028.65 & 987.00 & 329.56 & 178.00 & 539.87 & 476.00\\
& 560000 & 840000 & 7471.88 & 7397.00 & 7592.03 & 7491.00 & 2289.97 & 1825.00 & 2048.68 & 1718.00\\
& 5600000 & 8400000 & 18110.28 & 17509.00 & 18244.43 & 17995.00 & 13781.65 & 13633.00 & 14882.91 & 14797.00\\
Two blocks
& 2 & 560 & 99.95 & 53.00 & 175.83 & 127.00 & 16.81 & 2.00 & 109.59 & 59.00\\
& 20 & 5600 & 120.04 & 62.00 & 189.44 & 148.00 & 24.79 & 2.00 & 115.88 & 65.00\\
& 200 & 56000 & 1195.01 & 1144.00 & 1310.89 & 1246.00 & 24.93 & 2.00 & 112.56 & 56.00\\
& 2000 & 560000 & 3583.89 & 3179.00 & 3566.36 & 3305.00 & 16.37 & 4.00 & 188.45 & 71.00\\
& 20000 & 5600000 & 14093.35 & 13816.00 & 14272.21 & 14026.00 & 98.35 & 56.00 & 155.19 & 117.00\\
\hline
\end{tabular}
\end{tiny}
\end{table}

We test our algorithm on gather problems of varying degrees of
irregularity. Let $b, b>0$ be a chosen, average block size (in some
unit, here \mpiint). We have $p$ MPI processes, and use as fixed
gather root $r=\floor{p/2}$.  Our problems are as follows with names
indicating how block sizes are chosen for the processes.
\begin{description}
\item[Same:] For process $i$, $m_i=b$.
\item[Random:] Each $m_i$ is chosen uniformly at random in the range
  $[1,2b]$.
\item[Spikes:] Each $m_i$ is either $\rho b, \rho>1$ or $1$, chosen
  randomly with probability $1/\rho$ for each process $i$.
\item[Decreasing:] For process $i$, $m_i=\floor{\frac{2b(p-i)}{p}}+1$
\item[Alternating:] For even numbered processes, $m_i=b+\floor{b/2}$,
  for odd numbered processes $m_i=b-\floor{b/2}$.
\item[Two blocks:] All $m_i=0$, except $m_0=b$ and
  $m_{p-1}=b$.
\end{description}
These problem types, except for the last, specifically always have
$m_i>0$. This choice ensures that an implementation cannot take
advantage of not having to send empty blocks. We perform a series of
(weak scaling) experiments with $b=1,10,\ldots,10\,0000$; the total
problem size in each case is $m=\sum_{0\leq i<p}m_i$ and increasing
linearly with $p$ (except for the two blocks problems). For comparison
with \mpigather and for the padding performance
Guideline~(\ref{guide:padding}), the padded block size is $\max_{0\leq
  i<p}m_i$ and the total size $m'=p\max_{0\leq i<p}m_i$. For the
\textbf{spikes} problems, we have taken $\rho=5$.

In our experiments we perform 75 time measurements of each of the
collective operations with 10 initial, not timed, warmup calls, and
compute average and minimum times (the fastest completion time seen
over the 75 repetitions). For the average times we have not done any
outlier removal.  Before each measurement, MPI processes are
synchronized with the native \mpibarrier operation, and the running
time of a measurement is the time of the slowest process, which for
the gather operations is usually the root process.

Our first test system is a small InfiniBand cluster with 36 nodes each
consisting of two 8-core AMD Opteron 6134 processors running at
2.3GHz. The interconnect is a QDR InfiniBand MT26428. We have tried
the implementations with three different MPI libraries, namely
\jupiternecmpi, \jupitermvapichtwotwo and \jupiteropenmpitwoone using
the \gccversion compiler with -O3 optimization.  We present the
results in tabular form, see Table~\ref{tab:necmpi},
Table~\ref{tab:mvapich} and Table~\ref{tab:openmpi}. Running times are
in microseconds ($\mu s$).


Average and minimum, best observed time differ considerably (which may
be due to outliers), and comparison based on averages may not be
well-founded. Nevertheless, the results show the three library
implementations of the \mpigather and \mpigatherv operations to
(surprisingly) differ considerably in quality. For \mpigather, this
can best be seen for the \textbf{same} problem type, where the \necmpi
minimum time is about 9000 $\mu s$ and \mvapich at 31000 $\mu s$ with
\mvapich at 13000 $\mu s$ for the largest problem instance. Also for
the smaller problem sizes, the differences can be considerable. For
all three libraries, it is also clear that \mpigatherv is implemented
with a trivial algorithm compared to \mpigather; this makes
\mpigatherv an expensive operation for small problem sizes. On the
other hand, the algorithms used for \mpigather are not well chosen for
large problem instances, where for all three libraries, the simple,
direct to root implementations used for \mpigatherv perform
better. The trivial performance Guideline~(\ref{guide:regirreg}) is
violated in such cases.

For the smaller instances of the irregular problem types, all
libraries fail Guideline~(\ref{guide:padding}) with their \mpigatherv
implementations by large factors, whereas \tuwgatherv, except for the
\textbf{two blocks} problems, easily fulfill the padding guideline,
often by a considerable factor; the new \tuwgatherv implementation is
faster than the library implementations often by factors of 5 to more
than 20.  There are even cases where \tuwgatherv is faster than the
library \mpigather implementations (seen for the \textbf{same} block
size problem type).

\begin{table}
\caption{Results for \intelmpi, default algorithms settings,
  $p=100\times 16=1600$. Running times are in microseconds ($\mu
  s$). Numbers in red, bold font show violations of the performance
  guidelines, Guideline~(\ref{guide:regirreg}) in the \mpigather
  column, Guideline~(\ref{guide:padding}) in column \mpigatherv}
\label{tab:intelmpi100}
\begin{tiny}
\begin{tabular}{crrrrrrrrrr}
Problem & $m$ & $m'$ & 
\multicolumn{2}{c}{\mpigather} & 
\multicolumn{2}{c}{Guideline~(\ref{guide:padding})} & 
\multicolumn{2}{c}{\mpigatherv} &
\multicolumn{2}{c}{\tuwgatherv} \\
& & & 
\multicolumn{1}{c}{avg} & \multicolumn{1}{c}{min} & 
\multicolumn{1}{c}{avg} & \multicolumn{1}{c}{min} & 
\multicolumn{1}{c}{avg} & \multicolumn{1}{c}{min} & 
\multicolumn{1}{c}{avg} & \multicolumn{1}{c}{min} \\
\hline
Same
& 1600 & 1600 & 6785.78 & \color{red}{\textbf{4585.03}} & 7325.16 & 5573.03 & 3881.19 & 2460.96 & 51.04 & 33.86\\
& 16000 & 16000 & 7829.35 & \color{red}{\textbf{5223.99}} & 8657.35 & 6438.97 & 5567.75 & 3410.10 & 137.79 & 81.06\\
& 160000 & 160000 & 8149.66 & \color{red}{\textbf{6762.03}} & 9401.96 & 7140.87 & 3437.90 & 2695.08 & 334.08 & 297.07\\
& 1600000 & 1600000 & 21294.59 & \color{red}{\textbf{8141.99}} & 23851.03 & 10309.93 & 14069.02 & 3842.12 & 2253.27 & 2161.03\\
& 16000000 & 16000000 & 195922.00 & \color{red}{\textbf{190047.03}} & 195237.15 & 192334.89 & 194682.58 & 189079.05 & 22477.23 & 21020.89\\
Random
& 2439 & 3200 & 7334.39 & 5353.93 & 7927.66 & 5856.04 & 4103.67 & 2656.94 & 65.00 & 35.05\\
& 16738 & 32000 & 7821.15 & 6059.89 & 8902.53 & 6178.86 & 4529.37 & 2595.90 & 108.56 & 62.94\\
& 162498 & 320000 & 10918.63 & 7423.16 & 12428.28 & 7917.88 & 4069.03 & 2729.89 & 337.99 & 308.04\\
& 1590753 & 3200000 & 38088.87 & 10705.95 & 39633.42 & 15350.10 & 24666.81 & 4288.91 & 2485.42 & 2177.95\\
& 15960523 & 31998400 & 47698.98 & 44461.97 & 48081.72 & 44685.84 & 283416.54 & 272669.08 & 23039.06 & 21776.91\\
Spikes
& 2908 & 8000 & 6946.98 & 5175.11 & 8367.50 & 5854.85 & 3964.65 & 2265.93 & 74.26 & 42.20\\
& 17476 & 80000 & 7544.19 & 5806.92 & 8502.44 & 6906.99 & 3722.79 & 2291.92 & 100.25 & 70.81\\
& 150302 & 800000 & 12313.47 & 7770.06 & 14366.38 & 7728.10 & 3464.65 & 2391.10 & 317.45 & 267.03\\
& 1596281 & 8000000 & 87790.86 & 83186.15 & 89378.70 & 86223.13 & 48634.55 & 3810.17 & 2166.81 & 2092.12\\
& 16051279 & 80000000 & 113636.35 & 105054.14 & 114928.19 & 104699.85 & 24557.99 & 22396.80 & 24263.74 & 22591.11\\
Decreasing
& 2402 & 4800 & 7178.11 & 4694.94 & 8533.26 & 6350.99 & 3950.61 & 2315.04 & 69.47 & 36.00\\
& 16820 & 33600 & 7104.90 & 4496.10 & 7834.66 & 6560.80 & 3602.98 & 2018.93 & 132.64 & 85.12\\
& 161000 & 321600 & 10067.75 & 6834.98 & 12254.34 & 7071.97 & 5158.21 & 2415.90 & 405.28 & 351.91\\
& 1602000 & 3201600 & 32169.61 & 10272.03 & 39441.70 & 14613.15 & 23314.90 & 3845.93 & 2682.00 & 2382.04\\
& 16011200 & 32001600 & 49304.18 & 44602.87 & 49702.01 & 45029.88 & 275515.48 & 30343.06 & 27429.17 & 23132.80\\
Alternating
& 1600 & 1600 & 7076.79 & 4745.01 & 7580.29 & 5074.02 & 3463.82 & 2285.96 & 59.53 & 31.95\\
& 16000 & 24000 & 7654.16 & 5013.94 & 7520.49 & 5375.86 & 4707.53 & 2198.93 & 78.36 & 59.13\\
& 160000 & 240000 & 8806.64 & 6551.98 & 10412.87 & 6759.88 & 3286.97 & 2158.88 & 334.26 & 303.03\\
& 1600000 & 2400000 & 25057.63 & 9432.08 & 26998.28 & 13028.86 & 15485.17 & 3649.95 & 2368.45 & 2177.00\\
& 16000000 & 24000000 & 293720.20 & 289201.97 & 294498.66 & 291086.91 & 273824.64 & 267002.11 & 21879.65 & 20879.03\\
Two blocks
& 2 & 1600 & 6305.51 & 4293.92 & 7515.37 & 5100.01 & 7.02 & 5.01 & 40.17 & 21.22\\
& 20 & 16000 & 7039.35 & 4035.95 & 7625.26 & 5104.06 & 11.43 & 5.01 & 40.76 & 22.89\\
& 200 & 160000 & 7868.89 & 6004.10 & 9141.56 & 7048.13 & 36.95 & 5.01 & 58.23 & 25.03\\
& 2000 & 1600000 & 22172.74 & 8018.97 & 25377.00 & 10148.05 & 14.13 & 6.91 & 56.07 & 29.80\\
& 20000 & 16000000 & 193004.67 & 40503.98 & 199676.28 & 187478.07 & 57.81 & 28.13 & 108.19 & 57.94\\
\hline
\end{tabular}
\end{tiny}
\end{table}

\begin{table}
\caption{Results for \intelmpi, default algorithms settings,
  $p=200\times 16=3200$. Running times are in microseconds ($\mu
  s$). Numbers in red, bold font show violations of the performance
  guidelines, Guideline~(\ref{guide:regirreg}) in the \mpigather
  column, Guideline~(\ref{guide:padding}) in column \mpigatherv}
\label{tab:intelmpi200}
\begin{tiny}
\begin{tabular}{crrrrrrrrrr}
Problem & $m$ & $m'$ & 
\multicolumn{2}{c}{\mpigather} & 
\multicolumn{2}{c}{Guideline~(\ref{guide:padding})} & 
\multicolumn{2}{c}{\mpigatherv} &
\multicolumn{2}{c}{\tuwgatherv} \\
& & & 
\multicolumn{1}{c}{avg} & \multicolumn{1}{c}{min} & 
\multicolumn{1}{c}{avg} & \multicolumn{1}{c}{min} & 
\multicolumn{1}{c}{avg} & \multicolumn{1}{c}{min} & 
\multicolumn{1}{c}{avg} & \multicolumn{1}{c}{min} \\
\hline
Same
& 3200 & 3200 & 28437.24 & 19305.94 & 28689.85 & 21116.97 & 30901.47 & \color{red}{\textbf{23523.81}} & 92.55 & 41.01\\
& 32000 & 32000 & 27679.23 & 21731.85 & 28434.82 & 23573.16 & 32097.67 & \color{red}{\textbf{26045.08}} & 143.67 & 90.12\\
& 320000 & 320000 & 24762.15 & 19256.11 & 22451.83 & 18827.92 & 27518.91 & \color{red}{\textbf{25965.21}} & 593.66 & 538.11\\
& 3200000 & 3200000 & 44065.08 & 24388.07 & 55709.50 & 31090.97 & 47223.59 & 30979.87 & 4813.63 & 4436.02\\
& 32000000 & 32000000 & 300888.95 & 218910.93 & 323063.74 & 227644.92 & 308667.72 & 220571.99 & 48447.10 & 45253.99\\
Random
& 4829 & 6400 & 24021.80 & 15135.05 & 24180.03 & 15941.86 & 22893.44 & 15330.08 & 108.13 & 64.85\\
& 33497 & 64000 & 21767.06 & 17243.15 & 22250.89 & 18190.86 & 24442.87 & 15856.98 & 182.17 & 108.00\\
& 321787 & 640000 & 28661.41 & 18878.22 & 36155.43 & 20350.93 & 27181.52 & \color{red}{\textbf{23756.98}} & 598.71 & 530.00\\
& 3169326 & 6400000 & 67277.72 & 44735.91 & 76511.79 & 48445.94 & 52865.61 & 39208.89 & 4665.05 & 4305.12\\
& 31446972 & 63996800 & 109308.41 & 99385.02 & 108454.98 & 95494.03 & 310909.81 & 272596.12 & 40920.26 & 40133.95\\
Spikes
& 5560 & 16000 & 23889.74 & 17380.95 & 24517.83 & 19174.10 & 26932.36 & \color{red}{\textbf{19545.08}} & 86.25 & 50.07\\
& 33335 & 160000 & 21551.87 & 18486.98 & 22186.02 & 18473.86 & 28139.84 & \color{red}{\textbf{19646.88}} & 140.92 & 97.04\\
& 326053 & 1600000 & 37614.00 & 21208.05 & 40370.15 & 23549.08 & 28761.65 & 23475.89 & 590.82 & 535.01\\
& 3142572 & 16000000 & 252393.25 & 93301.06 & 263011.72 & 99792.00 & 81848.83 & 80199.96 & 4185.84 & 3964.90\\
& 33252535 & 160000000 & 237369.09 & 218248.84 & 234426.21 & 219135.05 & 66427.58 & 60978.17 & 47125.30 & 44766.90\\
Decreasing
& 4802 & 9600 & 22828.26 & 16659.02 & 23782.15 & 16151.19 & 26215.08 & 19952.06 & 106.31 & 61.04\\
& 33620 & 67200 & 22501.45 & 18682.00 & 24277.93 & 18852.95 & 28763.90 & 21468.16 & 157.98 & 124.93\\
& 321800 & 643200 & 33628.73 & 19536.02 & 40671.71 & 20276.07 & 35724.42 & 33167.12 & 653.67 & 598.91\\
& 3202800 & 6403200 & 68307.70 & 44342.04 & 79130.70 & 47766.92 & 61027.90 & 36426.78 & 4953.15 & 4633.19\\
& 32012000 & 64003200 & 106778.59 & 97492.93 & 104682.51 & 99535.94 & 297412.84 & \color{red}{\textbf{282826.19}} & 47414.95 & 43900.97\\
Alternating
& 3200 & 3200 & 25059.58 & 15811.92 & 24361.71 & 16531.94 & 25777.54 & \color{red}{\textbf{19775.15}} & 103.63 & 49.83\\
& 32000 & 48000 & 22880.73 & 17456.05 & 24041.51 & 18349.89 & 25712.11 & \color{red}{\textbf{22413.02}} & 141.23 & 100.14\\
& 320000 & 480000 & 27044.46 & 18890.86 & 26782.88 & 18590.93 & 28739.11 & \color{red}{\textbf{25625.94}} & 625.87 & 550.99\\
& 3200000 & 4800000 & 54848.80 & 30683.99 & 62371.10 & 40065.05 & 51288.62 & 30776.98 & 4826.48 & 4387.86\\
& 32000000 & 48000000 & 315070.99 & 311951.16 & 317679.07 & 313787.94 & 307983.23 & 282216.07 & 47164.98 & 44992.92\\
Two blocks
& 2 & 3200 & 25678.38 & 15871.05 & 23928.10 & 13593.91 & 11.10 & 8.11 & 36.08 & 23.13\\
& 20 & 32000 & 23554.28 & 17744.06 & 22819.97 & 17996.07 & 12.20 & 8.11 & 65.89 & 25.03\\
& 200 & 320000 & 22381.84 & 17951.97 & 22820.16 & 19130.95 & 17.23 & 8.82 & 54.15 & 30.04\\
& 2000 & 3200000 & 46755.45 & 24815.08 & 55163.29 & 30699.97 & 18.10 & 8.82 & 88.52 & 42.92\\
& 20000 & 32000000 & 308257.55 & 210777.04 & 332478.58 & 218914.99 & 39.49 & 31.95 & 124.42 & 74.15\\
\hline
\end{tabular}
\end{tiny}
\end{table}

\begin{table}
\caption{Results for \intelmpi, default algorithms settings,
  $p=400\times 16=6400$. Running times are in microseconds ($\mu
  s$). Numbers in red, bold font show violations of the performance
  guidelines, Guideline~(\ref{guide:regirreg}) in the \mpigather
  column, Guideline~(\ref{guide:padding}) in column \mpigatherv}
\label{tab:intelmpi400-default}
\begin{tiny}
\begin{tabular}{crrrrrrrrrr}
Problem & $m$ & $m'$ & 
\multicolumn{2}{c}{\mpigather} & 
\multicolumn{2}{c}{Guideline~(\ref{guide:padding})} & 
\multicolumn{2}{c}{\mpigatherv} &
\multicolumn{2}{c}{\tuwgatherv} \\
& & & 
\multicolumn{1}{c}{avg} & \multicolumn{1}{c}{min} & 
\multicolumn{1}{c}{avg} & \multicolumn{1}{c}{min} & 
\multicolumn{1}{c}{avg} & \multicolumn{1}{c}{min} & 
\multicolumn{1}{c}{avg} & \multicolumn{1}{c}{min} \\
\hline
Same
& 6400 & 6400 & 143224.73 & 103827.95 & 134339.83 & 116875.89 & 198096.70 & \color{red}{\textbf{147197.96}} & 483.80 & 50.07 \\
& 64000 & 64000 & 127234.13 & 96213.10 & 119072.93 & 96263.17 & 157288.30 & \color{red}{\textbf{127314.09}} & 1119.38 & 158.07 \\
& 640000 & 640000 & 135538.09 & 92452.05 & 115266.33 & 98417.04 & 184500.61 & \color{red}{\textbf{150245.90}} & 1208.13 & 964.16 \\
& 6400000 & 6400000 & 174751.44 & 111608.03 & 168121.69 & 127385.85 & 229787.77 & \color{red}{\textbf{162585.97}} & 8883.16 & 8399.01 \\
& 64000000 & 64000000 & 478717.77 & 351228.00 & 448441.26 & 345794.92 & 521594.60 & 400845.05 & 82694.19 & 81154.11 \\
Random
& 9606 & 12800 & 112400.91 & 83522.08 & 119423.66 & 97676.99 & 165242.37 & \color{red}{\textbf{143667.94}} & 1149.91 & 80.11 \\
& 67404 & 128000 & 113700.98 & 92371.94 & 114324.38 & 100589.04 & 159724.97 & \color{red}{\textbf{135301.83}} & 746.22 & 169.04 \\
& 640316 & 1280000 & 130619.07 & 97826.96 & 117179.47 & 100124.84 & 196255.20 & \color{red}{\textbf{154004.10}} & 1431.84 & 959.16 \\
& 6365598 & 12800000 & 266579.81 & 133107.19 & 256535.67 & 194566.01 & 240661.89 & 174778.94 & 9395.70 & 8400.92 \\
& 64366469 & 128000000 & 383281.22 & 240486.86 & 382698.89 & 234140.87 & 568810.39 & 449437.14 & 85171.69 & 84529.88 \\
Spikes
& 11576 & 32000 & 115877.39 & 83566.90 & 107814.19 & 95278.98 & 155785.61 & \color{red}{\textbf{128391.03}} & 1546.82 & 67.00 \\
& 68483 & 320000 & 112930.08 & 88580.85 & 106333.22 & 95693.83 & 158825.65 & \color{red}{\textbf{129721.16}} & 912.51 & 162.12 \\
& 641627 & 3200000 & 186607.47 & 99061.97 & 179326.41 & 102573.87 & 170716.86 & \color{red}{\textbf{143610.00}} & 1132.60 & 957.97 \\
& 6465108 & 32000000 & 363413.76 & 282375.81 & 352513.17 & 273236.04 & 228685.47 & 170839.07 & 9048.14 & 8200.88 \\
& 64505110 & 320000000 & 592158.24 & 477687.12 & 589359.86 & 469671.96 & 271884.27 & 195724.96 & 139265.59 & 138741.97 \\
Decreasing
& 9602 & 19200 & 112524.48 & 86571.93 & 114034.30 & 97607.14 & 166033.34 & \color{red}{\textbf{143018.01}} & 372.60 & 56.98 \\
& 67220 & 134400 & 113499.48 & 88942.05 & 111749.48 & 99705.93 & 169566.57 & \color{red}{\textbf{137952.09}} & 297.82 & 197.89 \\
& 643400 & 1286400 & 130733.47 & 100178.96 & 122006.55 & 96143.96 & 204257.36 & \color{red}{\textbf{170412.06}} & 1309.55 & 1060.96 \\
& 6404400 & 12806400 & 256993.15 & 163550.14 & 249910.52 & 174806.12 & 236896.72 & \color{red}{\textbf{177741.05}} & 9904.83 & 9175.06 \\
& 64013600 & 128006400 & 358370.11 & 231245.99 & 381372.81 & 230093.00 & 531949.64 & 438292.98 & 87494.51 & 86869.00 \\
Alternating
& 6400 & 6400 & 103188.01 & 73747.87 & 103039.74 & 90380.19 & 155132.28 & \color{red}{\textbf{132302.05}} & 200.55 & 47.92 \\
& 64000 & 96000 & 108217.47 & 81782.10 & 104217.22 & 80809.83 & 144754.90 & \color{red}{\textbf{114871.03}} & 285.87 & 151.87 \\
& 640000 & 960000 & 128934.07 & 87894.20 & 108839.05 & 91517.93 & 189208.25 & \color{red}{\textbf{128816.84}} & 1136.50 & 955.10 \\
& 6400000 & 9600000 & 209340.04 & 124116.18 & 210117.61 & 134249.93 & 251695.84 & \color{red}{\textbf{152398.11}} & 9099.22 & 8328.91 \\
& 64000000 & 96000000 & 515877.56 & 367428.06 & 516283.90 & 398251.06 & 605431.46 & \color{red}{\textbf{453334.81}} & 81952.00 & 81115.96 \\
Two blocks
& 2 & 6400 & 102869.81 & 77906.85 & 99906.61 & 88593.96 & 17.60 & 15.97 & 365.01 & 30.99 \\
& 20 & 64000 & 105361.90 & 76925.99 & 104937.77 & 83439.11 & 17.81 & 15.97 & 56.81 & 30.04 \\
& 200 & 640000 & 107798.00 & 88968.99 & 105016.54 & 91776.85 & 17.78 & 16.93 & 146.50 & 29.80 \\
& 2000 & 6400000 & 184417.88 & 108457.09 & 179452.45 & 118996.14 & 21.39 & 16.93 & 165.49 & 36.00 \\
& 20000 & 64000000 & 435541.00 & 252403.02 & 416582.60 & 335694.07 & 38.32 & 30.04 & 83.75 & 63.90 \\
\hline
\end{tabular}
\end{tiny}
\end{table}

Our second system is a medium-large InfiniBand/Intel cluster
consisting of 2000 Dual Intel Xeon E5-2650v2 8-core processors running
at 2.6GHz, interconnected with an InfiniBand QDR-80
network\footnote{This is the so-called Vienna Scientific Cluster, see
  \url{vsc.ac.at}. We thank for access to this machine.}.  The MPI
library is \vscintelmpi and the compiler is \vscintelmpi with
optimization level -O3. The benchmark was executed with the default
environment for this machine, which pins the MPI processes to the 16
cores per node; also the choice of \mpigather and \mpigatherv
implementations was left to the environment. Results can be found in
Table~\ref{tab:intelmpi100}, Table~\ref{tab:intelmpi200} and
Table~\ref{tab:intelmpi400-default} for $p=1600$, $p=3200$, and
$p=6400$, respectively.  The most conspicuous observation about this
MPI library is the poor quality of both \mpigather and \mpigatherv,
rendering our \tuwgatherv implementation several orders of magnitude
faster for small problems. The \tuwgatherv implementation therefore
satisfies Guideline~(\ref{guide:padding}) by a large margin, and
\mpigather fails the trivial Guideline~(\ref{guide:regirreg}) compared
to \tuwgatherv by a very large factor.

\begin{table}
\caption{Results for \intelmpi, $p=400\times 16=6400$, with 1)
  binomial \mpigather (1) and linear \mpigatherv (1). Running times
  are in microseconds ($\mu s$)}
\label{tab:intelmpi400-1}
\begin{tiny}
\begin{tabular}{crrrrrrrrrr}
Problem & $m$ & $m'$ & 
\multicolumn{2}{c}{\mpigather} & 
\multicolumn{2}{c}{Guideline~(\ref{guide:padding})} & 
\multicolumn{2}{c}{\mpigatherv} &
\multicolumn{2}{c}{\tuwgatherv} \\
& & & 
\multicolumn{1}{c}{avg} & \multicolumn{1}{c}{min} & 
\multicolumn{1}{c}{avg} & \multicolumn{1}{c}{min} & 
\multicolumn{1}{c}{avg} & \multicolumn{1}{c}{min} & 
\multicolumn{1}{c}{avg} & \multicolumn{1}{c}{min} \\
\hline
Same
& 6400 & 6400 & 286.35 & 35.05 & 317.71 & 64.13 & 203847.75 & 171184.06 & 299.52 & 52.21 \\
& 64000 & 64000 & 299.46 & 182.15 & 517.45 & 246.05 & 186710.49 & 153241.87 & 200.16 & 154.02 \\
& 640000 & 640000 & 1259.86 & 1109.84 & 1247.51 & 1140.12 & 208955.90 & 174618.96 & 1137.10 & 953.91 \\
& 6400000 & 6400000 & 13873.83 & 13394.12 & 14012.66 & 13395.07 & 272070.22 & 187475.92 & 8286.86 & 8234.02 \\
& 64000000 & 64000000 & 147270.64 & 144801.14 & 147243.89 & 144757.03 & 573076.70 & 451709.99 & 82321.16 & 80964.80 \\
Random
& 9648 & 12800 & 171.47 & 50.07 & 257.15 & 80.11 & 181675.47 & 148651.84 & 365.80 & 61.04 \\
& 67611 & 128000 & 438.32 & 299.93 & 451.96 & 334.02 & 164536.88 & 143263.10 & 192.52 & 166.89 \\
& 644099 & 1280000 & 2375.90 & 2069.95 & 2211.75 & 2100.94 & 180425.41 & 140429.02 & 1009.34 & 962.97 \\
& 6348958 & 12800000 & 27995.71 & 27446.99 & 28207.40 & 27549.98 & 218233.53 & 153713.94 & 8429.70 & 8240.94 \\
& 64822758 & 127993600 & 295035.76 & 290700.91 & 294802.43 & 290601.02 & 564158.68 & 449219.94 & 86925.19 & 84053.99 \\
Spikes
& 11668 & 32000 & 284.55 & 77.96 & 214.27 & 108.00 & 158854.36 & 130697.01 & 97.69 & 62.94 \\
& 70149 & 320000 & 821.23 & 600.10 & 897.93 & 660.18 & 163111.67 & 144008.16 & 183.77 & 164.99 \\
& 627655 & 3200000 & 5825.37 & 5667.21 & 5983.96 & 5713.94 & 156893.44 & 135225.06 & 1031.90 & 931.02 \\
& 6590083 & 32000000 & 70683.11 & 69520.00 & 72239.96 & 69606.07 & 231074.29 & 170685.05 & 8772.37 & 8465.05 \\
& 63605128 & 320000000 & 720784.90 & 710308.07 & 724077.12 & 710829.02 & 237647.31 & 186805.01 & 82233.83 & 81801.89 \\
Decreasing
& 9602 & 19200 & 79.90 & 45.06 & 332.72 & 87.02 & 166301.72 & 124325.04 & 89.52 & 70.10 \\
& 67220 & 134400 & 437.02 & 314.95 & 580.22 & 342.85 & 169710.59 & 131918.91 & 315.33 & 205.04 \\
& 643400 & 1286400 & 2177.98 & 2104.04 & 2323.18 & 2146.96 & 201109.51 & 154931.78 & 1102.36 & 1046.18 \\
& 6404400 & 12806400 & 27142.07 & 26774.88 & 27319.14 & 26860.00 & 260218.10 & 183459.04 & 9517.38 & 9062.05 \\
& 64013600 & 128006400 & 291315.60 & 286717.89 & 290913.77 & 286838.05 & 531015.30 & 420577.05 & 87383.89 & 86689.95 \\
Alternating
& 6400 & 6400 & 49.76 & 37.91 & 212.76 & 64.85 & 161015.64 & 138704.06 & 76.93 & 50.07 \\
& 64000 & 96000 & 252.95 & 231.98 & 305.06 & 262.98 & 156126.28 & 125625.13 & 172.09 & 153.06 \\
& 640000 & 960000 & 1693.25 & 1624.11 & 2019.42 & 1688.00 & 184365.52 & 147700.07 & 987.06 & 964.88 \\
& 6400000 & 9600000 & 17999.47 & 17704.01 & 18028.99 & 17511.13 & 225913.18 & 155328.99 & 8507.52 & 8383.99 \\
& 64000000 & 96000000 & 217804.71 & 214792.01 & 217947.53 & 214887.14 & 626248.70 & 532054.90 & 82007.54 & 81188.92 \\
Two blocks
& 2 & 6400 & 51.97 & 30.99 & 82.22 & 61.99 & 18.67 & 15.97 & 179.34 & 40.05 \\
& 20 & 64000 & 606.29 & 180.96 & 357.92 & 218.87 & 17.78 & 16.93 & 50.11 & 30.04 \\
& 200 & 640000 & 1170.90 & 1126.05 & 1219.57 & 1170.87 & 17.78 & 15.97 & 42.36 & 29.80 \\
& 2000 & 6400000 & 10994.00 & 10734.80 & 11114.82 & 10782.96 & 20.72 & 16.93 & 153.94 & 37.91 \\
& 20000 & 64000000 & 141559.13 & 139204.03 & 142086.64 & 139410.97 & 36.68 & 25.99 & 112.17 & 83.92 \\
\hline
\end{tabular}
\end{tiny}
\end{table}

\begin{table}
\caption{Results for \intelmpi, $p=400\times 16=6400$, with topology
  aware binomial \mpigather (2) and topology aware linear \mpigatherv
  (2).  Running times are in microseconds ($\mu s$)}
\label{tab:intelmpi400-2}
\begin{tiny}
\begin{tabular}{crrrrrrrrrr}
Problem & $m$ & $m'$ & 
\multicolumn{2}{c}{\mpigather} & 
\multicolumn{2}{c}{Guideline~(\ref{guide:padding})} & 
\multicolumn{2}{c}{\mpigatherv} &
\multicolumn{2}{c}{\tuwgatherv} \\
& & & 
\multicolumn{1}{c}{avg} & \multicolumn{1}{c}{min} & 
\multicolumn{1}{c}{avg} & \multicolumn{1}{c}{min} & 
\multicolumn{1}{c}{avg} & \multicolumn{1}{c}{min} & 
\multicolumn{1}{c}{avg} & \multicolumn{1}{c}{min} \\
\hline
Same
& 6400 & 6400 & 1330.29 & 1194.00 & 1396.46 & 1277.92 & 189631.51 & 144186.02 & 66.71 & 50.07 \\
& 64000 & 64000 & 1525.18 & 1385.93 & 1580.72 & 1460.08 & 182839.87 & 170145.03 & 212.39 & 154.02 \\
& 640000 & 640000 & 2534.53 & 2250.91 & 2975.20 & 2544.88 & 181710.66 & 167731.05 & 1127.30 & 956.06 \\
& 6400000 & 6400000 & 19200.84 & 18810.03 & 37401.32 & 18851.04 & 215802.32 & 186728.00 & 8649.79 & 8425.95 \\
& 64000000 & 64000000 & 191392.45 & 190282.11 & 191377.15 & 190342.90 & 2188225.57 & 1535592.08 & 81870.86 & 81009.86 \\
Random
& 9568 & 12800 & 1347.94 & 1218.08 & 1411.15 & 1275.06 & 152857.23 & 120592.12 & 115.19 & 96.08 \\
& 67129 & 128000 & 1500.26 & 1363.99 & 1547.07 & 1474.14 & 150938.58 & 140668.15 & 306.75 & 165.94 \\
& 647126 & 1280000 & 3855.85 & 2815.01 & 3782.53 & 3623.01 & 148638.11 & 144186.02 & 1005.82 & 965.83 \\
& 6356994 & 12800000 & 52434.22 & 52086.83 & 52599.47 & 52096.84 & 212271.66 & 175087.21 & 8581.09 & 8364.92 \\
& 64199878 & 127993600 & 345662.00 & 344324.11 & 345808.56 & 344326.97 & 845814.00 & 610232.83 & 85244.80 & 84436.89 \\
Spikes
& 11336 & 32000 & 1360.41 & 1244.07 & 1507.41 & 1328.95 & 152146.13 & 139024.02 & 90.99 & 61.99 \\
& 67307 & 320000 & 1861.01 & 1482.96 & 1891.37 & 1816.99 & 147729.20 & 141247.03 & 176.17 & 161.89 \\
& 657595 & 3200000 & 10138.76 & 6142.85 & 10233.53 & 10045.05 & 148078.55 & 141868.11 & 1017.14 & 981.81 \\
& 6230155 & 32000000 & 107214.56 & 106601.00 & 107251.18 & 106725.93 & 225041.82 & 195419.07 & 8361.31 & 8239.98 \\
& 63355133 & 320000000 & 914865.86 & 745203.02 & 912313.57 & 745388.98 & 400093.68 & 373766.90 & 82541.50 & 81135.99 \\
Decreasing
& 9602 & 19200 & 1362.97 & 1241.92 & 1380.17 & 1311.06 & 158747.43 & 131989.00 & 79.90 & 56.98 \\
& 67220 & 134400 & 1504.44 & 1348.02 & 1591.09 & 1502.99 & 352687.02 & 154788.02 & 228.44 & 205.04 \\
& 643400 & 1286400 & 3759.18 & 3068.92 & 3755.07 & 3632.07 & 175958.64 & 168071.03 & 1118.99 & 1065.97 \\
& 6404400 & 12806400 & 53061.23 & 52588.94 & 53175.74 & 52703.86 & 229055.11 & 195079.80 & 9484.40 & 9130.00 \\
& 64013600 & 128006400 & 347145.11 & 345803.98 & 549690.12 & 345946.07 & 862627.21 & 505980.01 & 87318.44 & 86655.14 \\
Alternating
& 6400 & 6400 & 1302.95 & 1215.93 & 1318.34 & 1257.90 & 146806.19 & 131737.95 & 110.06 & 95.84 \\
& 64000 & 96000 & 1432.50 & 1291.99 & 1484.10 & 1376.15 & 151094.88 & 139039.99 & 175.84 & 153.06 \\
& 640000 & 960000 & 2990.64 & 1932.14 & 3073.82 & 2930.88 & 147378.10 & 142193.08 & 998.86 & 972.03 \\
& 6400000 & 9600000 & 41663.82 & 41385.89 & 41753.82 & 41460.99 & 184450.35 & 164529.09 & 8700.32 & 8590.94 \\
& 64000000 & 96000000 & 270850.54 & 269791.13 & 270583.79 & 269803.05 & 2210973.20 & 1426429.03 & 82099.65 & 81186.06 \\
Two blocks
& 2 & 6400 & 1308.91 & 1210.93 & 1332.94 & 1256.94 & 37.09 & 36.00 & 43.52 & 28.85 \\
& 20 & 64000 & 1391.42 & 1288.89 & 1475.43 & 1389.98 & 36.52 & 35.05 & 102.57 & 45.06 \\
& 200 & 640000 & 2377.34 & 1743.79 & 2432.59 & 2389.91 & 36.95 & 36.00 & 57.69 & 31.95 \\
& 2000 & 6400000 & 18490.39 & 17412.19 & 18543.29 & 18399.95 & 38.34 & 36.00 & 51.77 & 36.95 \\
& 20000 & 64000000 & 191868.46 & 191071.99 & 191871.42 & 190980.91 & 53.24 & 41.96 & 181.25 & 63.18 \\
\hline
\end{tabular}
\end{tiny}
\end{table}

\begin{table}
\caption{Results for \intelmpi, $p=400\times 16=6400$, with Shumilin's
  algorithm for \mpigather (3) and $k$-nomial \mpigatherv (3). Running
  times are in microseconds ($\mu s$)}
\label{tab:intelmpi400-3}
\begin{tiny}
\begin{tabular}{crrrrrrrrrr}
Problem & $m$ & $m'$ & 
\multicolumn{2}{c}{\mpigather} & 
\multicolumn{2}{c}{Guideline~(\ref{guide:padding})} & 
\multicolumn{2}{c}{\mpigatherv} &
\multicolumn{2}{c}{\tuwgatherv} \\
& & & 
\multicolumn{1}{c}{avg} & \multicolumn{1}{c}{min} & 
\multicolumn{1}{c}{avg} & \multicolumn{1}{c}{min} & 
\multicolumn{1}{c}{avg} & \multicolumn{1}{c}{min} & 
\multicolumn{1}{c}{avg} & \multicolumn{1}{c}{min} \\
\hline
Same
& 6400 & 6400 & 133951.06 & 100253.11 & 129760.14 & 115872.86 & 146.20 & 118.02 & 85.65 & 50.07 \\
& 64000 & 64000 & 138576.78 & 104353.90 & 135104.30 & 119436.03 & 299.16 & 273.94 & 169.00 & 153.06 \\
& 640000 & 640000 & 145147.18 & 110294.82 & 135611.39 & 121351.00 & 2015.99 & 1956.94 & 983.86 & 960.11 \\
& 6400000 & 6400000 & 207961.37 & 134063.01 & 218410.02 & 149597.17 & 19388.81 & 18954.99 & 8535.40 & 8424.04 \\
& 64000000 & 64000000 & 523186.45 & 379465.82 & 522073.89 & 421118.97 & 193644.12 & 191694.02 & 81971.96 & 81095.93 \\
Random
& 9589 & 12800 & 127994.05 & 88747.02 & 118543.32 & 93716.86 & 171.78 & 159.03 & 93.98 & 65.09 \\
& 67181 & 128000 & 113890.80 & 90431.93 & 117950.88 & 101103.07 & 395.44 & 337.12 & 191.22 & 170.95 \\
& 644560 & 1280000 & 122207.36 & 97154.14 & 111394.18 & 95774.17 & 2036.46 & 1990.08 & 1017.37 & 961.78 \\
& 6396993 & 12800000 & 238959.09 & 158728.12 & 233639.79 & 165606.02 & 19961.73 & 19591.09 & 8562.66 & 8450.98 \\
& 64626032 & 127987200 & 374623.55 & 231142.04 & 387951.31 & 232944.97 & 198733.35 & 196699.14 & 84133.59 & 83119.87 \\
Spikes
& 11408 & 32000 & 100735.14 & 78346.01 & 106124.63 & 81103.09 & 162.15 & 133.04 & 101.20 & 72.00 \\
& 67356 & 320000 & 110304.74 & 86798.91 & 106089.93 & 88222.03 & 381.43 & 302.08 & 182.23 & 160.93 \\
& 644621 & 3200000 & 172453.25 & 92860.94 & 157690.17 & 100098.13 & 2064.85 & 2034.90 & 1000.13 & 915.05 \\
& 6250151 & 32000000 & 358345.68 & 231605.05 & 334471.33 & 269568.92 & 20139.45 & 19890.07 & 8674.33 & 8198.02 \\
& 64605108 & 320000000 & 587072.46 & 469606.16 & 581391.97 & 464993.95 & 197733.14 & 195786.95 & 84086.85 & 83271.03 \\
Decreasing
& 9602 & 19200 & 112805.33 & 79298.97 & 106625.11 & 89558.12 & 445.74 & 120.88 & 489.41 & 56.98 \\
& 67220 & 134400 & 102060.34 & 82747.94 & 105274.30 & 89476.11 & 678.47 & 325.92 & 552.13 & 205.99 \\
& 643400 & 1286400 & 117084.49 & 89455.84 & 107227.53 & 88491.92 & 2463.06 & 2162.93 & 1471.62 & 1075.98 \\
& 6404400 & 12806400 & 236196.07 & 134517.19 & 236808.20 & 162189.96 & 21575.03 & 21152.97 & 9372.68 & 9088.99 \\
& 64013600 & 128006400 & 381298.96 & 237378.12 & 380227.12 & 234842.06 & 207016.92 & 204067.95 & 88299.37 & 86812.02 \\
Alternating
& 6400 & 6400 & 109061.30 & 78548.91 & 102702.71 & 90834.14 & 946.29 & 126.12 & 383.27 & 52.93 \\
& 64000 & 96000 & 110752.97 & 75751.07 & 106155.93 & 86133.96 & 321.59 & 282.05 & 659.15 & 151.87 \\
& 640000 & 960000 & 125573.36 & 86350.92 & 110846.78 & 97111.94 & 2102.71 & 1987.22 & 1544.37 & 958.92 \\
& 6400000 & 9600000 & 197330.75 & 108726.02 & 187050.61 & 139098.88 & 19757.82 & 19280.91 & 8888.04 & 8229.97 \\
& 64000000 & 96000000 & 570109.35 & 357120.99 & 625595.87 & 395962.00 & 193499.81 & 191794.16 & 82379.72 & 81126.93 \\
Two blocks
& 2 & 6400 & 106856.36 & 88028.91 & 108394.58 & 89107.04 & 315.90 & 93.94 & 463.92 & 30.04 \\
& 20 & 64000 & 111063.87 & 79486.85 & 105336.73 & 80888.99 & 410.45 & 92.03 & 679.39 & 30.04 \\
& 200 & 640000 & 116336.94 & 86739.06 & 105976.77 & 90575.22 & 425.60 & 94.89 & 369.27 & 36.00 \\
& 2000 & 6400000 & 189806.36 & 105443.00 & 158316.22 & 118196.01 & 416.14 & 103.00 & 356.42 & 35.05 \\
& 20000 & 64000000 & 480931.78 & 349302.05 & 461774.30 & 348691.94 & 670.28 & 261.78 & 702.30 & 66.04 \\
\hline
\end{tabular}
\end{tiny}
\end{table}

\begin{table}
\caption{Results for \intelmpi, $p=400\times 16=6400$, with binomial
  with segmentation \mpigather (4) and $k$-nomial \mpigatherv
  (3). Running times are in microseconds ($\mu s$)}
\label{tab:intelmpi400-4}
\begin{tiny}
\begin{tabular}{crrrrrrrrrr}
Problem & $m$ & $m'$ & 
\multicolumn{2}{c}{\mpigather} & 
\multicolumn{2}{c}{Guideline~(\ref{guide:padding})} & 
\multicolumn{2}{c}{\mpigatherv} &
\multicolumn{2}{c}{\tuwgatherv} \\
& & & 
\multicolumn{1}{c}{avg} & \multicolumn{1}{c}{min} & 
\multicolumn{1}{c}{avg} & \multicolumn{1}{c}{min} & 
\multicolumn{1}{c}{avg} & \multicolumn{1}{c}{min} & 
\multicolumn{1}{c}{avg} & \multicolumn{1}{c}{min} \\
\hline
Same
& 6400 & 6400 & 774.14 & 39.10 & 1182.02 & 65.09 & 995.37 & 123.98 & 680.45 & 53.17 \\
& 64000 & 64000 & 563.37 & 184.06 & 559.83 & 210.05 & 902.68 & 259.16 & 819.50 & 157.83 \\
& 640000 & 640000 & 1594.75 & 1112.94 & 1702.48 & 1157.05 & 2281.50 & 1932.86 & 1619.93 & 962.02 \\
& 6400000 & 6400000 & 14464.53 & 13731.96 & 14634.84 & 13547.18 & 19681.53 & 19338.13 & 9631.22 & 8454.08 \\
& 64000000 & 64000000 & 145706.98 & 143640.04 & 146483.18 & 144027.95 & 193182.23 & 191246.99 & 83206.43 & 81207.99 \\
Random
& 9618 & 12800 & 891.74 & 46.01 & 428.57 & 80.11 & 754.59 & 131.13 & 750.79 & 61.04 \\
& 67038 & 128000 & 669.99 & 291.11 & 570.12 & 329.97 & 1150.92 & 330.92 & 463.82 & 166.18 \\
& 642773 & 1280000 & 2374.02 & 2070.90 & 2569.19 & 2115.01 & 2226.74 & 1948.12 & 1330.13 & 972.99 \\
& 6390881 & 12800000 & 24897.40 & 24209.02 & 25154.10 & 24286.03 & 20149.21 & 19526.00 & 10030.93 & 9416.10 \\
& 63795923 & 127974400 & 275344.53 & 270759.82 & 275305.34 & 271631.00 & 194914.27 & 193742.99 & 87632.58 & 86949.83 \\
Spikes
& 11320 & 32000 & 589.96 & 69.86 & 686.96 & 108.00 & 448.54 & 144.96 & 909.52 & 62.94 \\
& 69316 & 320000 & 1136.18 & 616.07 & 1269.21 & 649.21 & 626.34 & 305.18 & 580.06 & 161.89 \\
& 648613 & 3200000 & 6597.36 & 5619.05 & 6944.84 & 5648.85 & 2524.01 & 1981.02 & 1662.80 & 938.89 \\
& 6420117 & 32000000 & 77553.25 & 76498.03 & 77929.93 & 76401.95 & 20666.20 & 20349.98 & 9497.62 & 8433.82 \\
& 65005100 & 320000000 & 671322.28 & 661470.89 & 672777.57 & 661302.09 & 195790.83 & 193390.13 & 86017.73 & 84006.07 \\
Decreasing
& 9602 & 19200 & 796.97 & 58.89 & 774.79 & 82.02 & 1077.95 & 123.02 & 290.08 & 55.07 \\
& 67220 & 134400 & 622.58 & 311.14 & 1018.03 & 348.09 & 899.55 & 332.83 & 713.06 & 201.94 \\
& 643400 & 1286400 & 2773.77 & 2087.12 & 2877.29 & 2130.03 & 2821.24 & 2244.00 & 1633.08 & 1086.00 \\
& 6404400 & 12806400 & 24904.33 & 23913.86 & 25353.75 & 23975.85 & 22353.53 & 21329.88 & 9713.58 & 9071.11 \\
& 64013600 & 128006400 & 274205.81 & 269502.88 & 275449.98 & 270691.16 & 205636.62 & 202884.91 & 87562.04 & 86633.92 \\
Alternating
& 6400 & 6400 & 895.98 & 44.82 & 962.86 & 64.13 & 341.39 & 118.97 & 806.74 & 51.02 \\
& 64000 & 96000 & 950.16 & 239.13 & 751.48 & 266.08 & 889.07 & 271.08 & 784.63 & 154.97 \\
& 640000 & 960000 & 1894.59 & 1615.05 & 2300.61 & 1645.80 & 2490.22 & 1931.91 & 1352.45 & 967.03 \\
& 6400000 & 9600000 & 18867.64 & 17627.00 & 19090.05 & 17637.97 & 20032.64 & 19195.08 & 8487.01 & 8227.83 \\
& 64000000 & 96000000 & 209066.01 & 206285.00 & 209464.47 & 206187.01 & 193477.21 & 191688.06 & 82237.05 & 80953.12 \\
Two blocks
& 2 & 6400 & 456.90 & 36.00 & 1026.37 & 62.94 & 627.13 & 93.94 & 673.90 & 29.80 \\
& 20 & 64000 & 988.18 & 190.02 & 665.18 & 215.05 & 532.54 & 103.95 & 592.47 & 39.82 \\
& 200 & 640000 & 1600.52 & 1129.15 & 1671.33 & 1173.02 & 430.17 & 101.09 & 471.95 & 29.80 \\
& 2000 & 6400000 & 11553.79 & 10770.08 & 12114.46 & 10823.97 & 433.90 & 107.05 & 581.97 & 38.86 \\
& 20000 & 64000000 & 143116.56 & 141374.83 & 144352.72 & 141484.98 & 607.52 & 234.13 & 581.15 & 63.90 \\
\hline
\end{tabular}
\end{tiny}
\end{table}

The reason for the poor performance on the larger cluster may be that
the \mpigather and \mpigatherv implementations used by default are ill
chosen.  The \vscintelmpi library indeed contains different algorithms
and implementations that can be chosen by environment variables
(\intelmpiadjust); it can also be controlled for which message and
process ranges particular implementations shall be used. For
\mpigather, the library lists four different algorithms (no specific
references are given), namely 1) binomial, 2) topology aware binomial,
3) a socalled Shumilin's algorithm, and 4) binomial with
segmentation. For \mpigatherv, three choices are possible, namely 1)
linear, 2) topology aware linear, and 3) $k$-nomial (with radix
$k=2$). For completeness, we ran the benchmark with all these explicit
choices as well, each for the full problem size range. The results are
given in Tables~\ref{tab:intelmpi400-1}
to~\ref{tab:intelmpi400-4}. Performance guideline violations are too
numerous and not shown in these tables, and for most of these
implementation choices, performance is excessively poor compared to
the \tuwgatherv implementation. The best, and most promising choices
are 1) binomial for \mpigather, and 3) $k$-nomial for
\mpigatherv. These results are shown together in
Table~\ref{tab:intelmpi400-best}, and performance guidelines
violations are marked in red as in the previous tables. The results
show that while a binomial tree algorithm can work well also for
\mpigatherv for almost regular problems, there is a penalty as the
problems get more irregular, which the \tuwgatherv implementation does
not have to pay, still being often faster than the native \mpigather
implementation. The \tuwgatherv implementation in all cases
outperforms $k$-nomial \mpigatherv for the \vscintelmpi library by a
factor of two to three (noteworthy also in the ``Two blocks''
distribution where a linear algorithm performs best).  The violations
of performance Guideline~(\ref{guide:padding}) by \mpigatherv are
surprising, and likely due to overhead in determining block sizes to
be used in the $k$-nomial tree. The \tuwgatherv implementation has no
guideline violations, even when compared to the good \vscintelmpi
library implementations.

\begin{table}
\caption{Results for \intelmpi, $p=400\times 16=6400$, overall best
  settings with binomial (1) for \mpigather and $k$-nomial (3) for
  \mpigatherv.  Running times are in microseconds ($\mu s$). Numbers
  in red, bold font show violations of the performance guidelines,
  Guideline~(\ref{guide:regirreg}) in the \mpigather column,
  Guideline~(\ref{guide:padding}) in column \mpigatherv}
\label{tab:intelmpi400-best}
\begin{tiny}
\begin{tabular}{crrrrrrrrrr}
Problem & $m$ & $m'$ & 
\multicolumn{2}{c}{\mpigather} & 
\multicolumn{2}{c}{Guideline~(\ref{guide:padding})} & 
\multicolumn{2}{c}{\mpigatherv} &
\multicolumn{2}{c}{\tuwgatherv} \\
& & & 
\multicolumn{1}{c}{avg} & \multicolumn{1}{c}{min} & 
\multicolumn{1}{c}{avg} & \multicolumn{1}{c}{min} & 
\multicolumn{1}{c}{avg} & \multicolumn{1}{c}{min} & 
\multicolumn{1}{c}{avg} & \multicolumn{1}{c}{min} \\
\hline
Same
& 6400 & 6400 & 286.35 & 35.05 & 317.71 & 64.13 & 995.37 & \color{red}{\textbf{123.98}} & 680.45 & 53.17 \\
& 64000 & 64000 & 299.46 & 182.15 & 517.45 & 246.05 & 902.68 & \color{red}{\textbf{259.16}} & 819.50 & 157.83 \\
& 640000 & 640000 & 1259.86 & 1109.84 & 1247.51 & 1140.12 & 2281.50 & \color{red}{\textbf{1932.86}} & 1619.93 & 962.02 \\
& 6400000 & 6400000 & 13873.83 & 13394.12 & 14012.66 & 13395.07 & 19681.53 & \color{red}{\textbf{19338.13}} & 9631.22 & 8454.08 \\
& 64000000 & 64000000 & 147270.64 & 144801.14 & 147243.89 & 144757.03 & 193182.23 & 191246.99 & 83206.43 & 81207.99 \\
Random
& 9648 & 12800 & 171.47 & 50.07 & 257.15 & 80.11 & 754.59 & \color{red}{\textbf{131.13}} & 750.79 & 61.04 \\
& 67611 & 128000 & 438.32 & 299.93 & 451.96 & 334.02 & 1150.92 & 330.92 & 463.82 & 166.18 \\
& 644099 & 1280000 & 2375.90 & 2069.95 & 2211.75 & 2100.94 & 2226.74 & 1948.12 & 1330.13 & 972.99 \\
& 6348958 & 12800000 & 27995.71 & 27446.99 & 28207.40 & 27549.98 & 20149.21 & 19526.00 & 10030.93 & 9416.10 \\
& 64822758 & 127993600 & 295035.76 & 290700.91 & 294802.43 & 290601.02 & 194914.27 & 193742.99 & 87632.58 & 86949.83 \\
Spikes
& 11668 & 32000 & 284.55 & 77.96 & 214.27 & 108.00 & 448.54 & \color{red}{\textbf{144.96}} & 909.52 & 62.94 \\
& 70149 & 320000 & 821.23 & 600.10 & 897.93 & 660.18 & 626.34 & 305.18 & 580.06 & 161.89 \\
& 627655 & 3200000 & 5825.37 & 5667.21 & 5983.96 & 5713.94 & 2524.01 & 1981.02 & 1662.80 & 938.89 \\
& 6590083 & 32000000 & 70683.11 & 69520.00 & 72239.96 & 69606.07 & 20666.20 & 20349.98 & 9497.62 & 8433.82 \\
& 63605128 & 320000000 & 720784.90 & 710308.07 & 724077.12 & 710829.02 & 195790.83 & 193390.13 & 86017.73 & 84006.07 \\
Decreasing
& 9602 & 19200 & 79.90 & 45.06 & 332.72 & 87.02 & 1077.95 & \color{red}{\textbf{123.02}} & 290.08 & 55.07 \\
& 67220 & 134400 & 437.02 & 314.95 & 580.22 & 342.85 & 899.55 & 332.83 & 713.06 & 201.94 \\
& 643400 & 1286400 & 2177.98 & 2104.04 & 2323.18 & 2146.96 & 2821.24 & \color{red}{\textbf{2244.00}} & 1633.08 & 1086.00 \\
& 6404400 & 12806400 & 27142.07 & 26774.88 & 27319.14 & 26860.00 & 22353.53 & 21329.88 & 9713.58 & 9071.11 \\
& 64013600 & 128006400 & 291315.60 & 286717.89 & 290913.77 & 286838.05 & 205636.62 & 202884.91 & 87562.04 & 86633.92 \\
Alternating
& 6400 & 6400 & 49.76 & 37.91 & 212.76 & 64.85 & 341.39 & \color{red}{\textbf{118.97}} & 806.74 & 51.02 \\
& 64000 & 96000 & 252.95 & 231.98 & 305.06 & 262.98 & 889.07 & \color{red}{\textbf{271.08}} & 784.63 & 154.97 \\
& 640000 & 960000 & 1693.25 & 1624.11 & 2019.42 & 1688.00 & 2490.22 & \color{red}{\textbf{1931.91}} & 1352.45 & 967.03 \\
& 6400000 & 9600000 & 17999.47 & 17704.01 & 18028.99 & 17511.13 & 20032.64 & \color{red}{\textbf{19195.08}} & 8487.01 & 8227.83 \\
& 64000000 & 96000000 & 217804.71 & 214792.01 & 217947.53 & 214887.14 & 193477.21 & 191688.06 & 82237.05 & 80953.12 \\
Two blocks
& 2 & 6400 & 51.97 & 30.99 & 82.22 & 61.99 & 627.13 & \color{red}{\textbf{93.94}} & 673.90 & 29.80 \\
& 20 & 64000 & 606.29 & 180.96 & 357.92 & 218.87 & 532.54 & 103.95 & 592.47 & 39.82 \\
& 200 & 640000 & 1170.90 & 1126.05 & 1219.57 & 1170.87 & 430.17 & 101.09 & 471.95 & 29.80 \\
& 2000 & 6400000 & 10994.00 & 10734.80 & 11114.82 & 10782.96 & 433.90 & 107.05 & 581.97 & 38.86 \\
& 20000 & 64000000 & 141559.13 & 139204.03 & 142086.64 & 139410.97 & 607.52 & 234.13 & 581.15 & 63.90 \\
\hline
\end{tabular}
\end{tiny}
\end{table}

\section{Conclusion}

This paper described new, simple algorithms for performing irregular
gather and scatter operations as found in MPI in linear communication
time, a considerable improvement over both fixed, data oblivious
logarithmic depth trees and direct communication with the root.  An
experimental evaluation shows that the resulting implementation can,
especially for overall small problem instances be considerably faster
than current MPI library \mpigatherv implementations by large factors.
Our prototype implementations can readily be incorporated into
existing MPI libraries.

The tree construction technique of Lemma~\ref{lem:construction} can be
applied to other problems as well, for instance to construct good,
problem dependent trees for sparse reduction
operations~\cite{Traff10:neutralsparse}.

\bibliographystyle{plain}
\bibliography{parallel,traff}

\begin{thebibliography}{10}

\bibitem{Ben-MiledFortesEigenmannTaylor98}
Zina Ben{-}Miled, Jos{\'{e}} A.~B. Fortes, Rudolf Eigenmann, and Valerie~E.
  Taylor.
\newblock On the implementation of broadcast, scatter and gather in a
  heterogeneous architecture.
\newblock In {\em Thirty-First Annual Hawaii International Conference on System
  Sciences {(HICSS)}}, pages 216--225, 1998.

\bibitem{BhattPucciRanadeRosenberg93}
Sandeep~N. Bhatt, Geppino Pucci, Abhiram Ranade, and Arnold~L. Rosenberg.
\newblock Scattering and gathering messages in networks of processors.
\newblock {\em {IEEE} Transactions on Computers}, 42(8):938--949, 1993.

\bibitem{BoxerMiller04}
Laurence Boxer and Russ Miller.
\newblock Coarse grained gather and scatter operations with applications.
\newblock {\em {J}ournal of {P}arallel and {D}istributed {C}omputing},
  64(11):1297--1310, 2004.

\bibitem{Traff17:expected}
Alexandra Carpen-Amarie, Sascha Hunold, and Jesper~Larsson Tr{\"a}ff.
\newblock On expected and observed communication performance with {MPI} derived
  datatypes.
\newblock Submitted, 2017.

\bibitem{ChanHeimlichPurkayasthavandeGeijn07}
Ernie Chan, Marcel Heimlich, Avi Purkayastha, and Robert~A. van~de Geijn.
\newblock Collective communication: theory, practice, and experience.
\newblock {\em Concurrency and Computation: Practice and Experience},
  19(13):1749--1783, 2007.

\bibitem{CharlesFraigniaud93}
Henri-Pierre Charles and Pierre Fraigniaud.
\newblock Scheduling a scattering-gathering sequence on hypercubes.
\newblock {\em {P}arallel {P}rocessing {L}etters}, 3:29--42, 1993.

\bibitem{DichevRychovLastovetsky10}
Kiril Dichev, Vladimir Rychkov, and Alexey~L. Lastovetsky.
\newblock Two algorithms of irregular scatter/gather operations for
  heterogeneous platforms.
\newblock In {\em Recent Advances in the Message Passing Interface; 17th
  European {MPI} Users' Group Meeting ({EuroMPI})}, pages 289--293, 2010.

\bibitem{HattaShibusa00}
Jun{-}ichi Hatta and Susumu Shibusawa.
\newblock Scheduling algorithms for efficient gather operations in distributed
  heterogeneous systems.
\newblock In {\em Proceedings of the 2000 International Workshop on Parallel
  Processing {(ICPPW)}}, pages 173--180, 2000.

\bibitem{Traff16:autoguide}
Sascha Hunold, Alexandra Carpen-Amarie, Felix~Donatus L{\"u}bbe, and
  Jesper~Larsson Tr{\"a}ff.
\newblock Automatic verification of self-consistent {MPI} performance
  guidelines.
\newblock In {\em Euro-Par Parallel Processing}, volume 9833 of {\em Lecture
  Notes in Computer Science}, pages 433--446, 2016.

\bibitem{MPI-3.0}
{MPI Forum}.
\newblock {\em \textsf{MPI}: A Message-Passing Interface Standard. Version
  3.0}, September 21st 2012.
\newblock \url{www.mpi-forum.org}.

\bibitem{SaadSchultz89}
Youcef Saad and Martin~H. Schultz.
\newblock Data communication in parallel architectures.
\newblock {\em {P}arallel {C}omputing}, 11(2):131--150, 1989.

\bibitem{ShibusawaMakinoNimiyaHatta00}
Susumu Shibusawa, Hiroyuki Makino, Shigeki Nimiya, and Jun{-}ichi Hatta.
\newblock Scatter and gather operations on an asynchronous communication model.
\newblock In {\em Proceedings of the 2000 {ACM} Symposium on Applied Computing
  ({SAC})}, pages 685--691, 2000.

\bibitem{Traff04:gatscat}
Jesper~Larsson Tr{\"{a}}ff.
\newblock Hierarchical gather/scatter algorithms with graceful degradation.
\newblock In {\em 18th International Parallel and Distributed Processing
  Symposium {(IPDPS)}}, page~80, 2004.

\bibitem{Traff10:neutralsparse}
Jesper~Larsson Tr{\"a}ff.
\newblock Transparent neutral element elimination in {MPI} reduction
  operations.
\newblock In {\em Recent Advances in Message Passing Interface. 17th European
  {MPI} Users' Group Meeting}, volume 6305 of {\em Lecture Notes in Computer
  Science}, pages 275--284. Springer, 2010.

\bibitem{Traff16:typeprog}
Jesper~Larsson Tr{\"a}ff.
\newblock A library for advanced datatype programming.
\newblock In {\em 23rd European MPI Users' Group Meeting ({EuroMPI})}, pages
  98--107. ACM, 2016.

\bibitem{Traff10:selfcons}
Jesper~Larsson Tr{\"{a}}ff, William~D. Gropp, and Rajeev Thakur.
\newblock Self-consistent {MPI} performance guidelines.
\newblock {\em {IEEE} Transactions on Parallel and Distributed Systems},
  21(5):698--709, 2010.

\end{thebibliography}

\end{document}